\documentclass[sigconf]{acmart}

\renewcommand\footnotetextcopyrightpermission[1]{} 

\usepackage{slashbox}
\usepackage{multirow}
\usepackage{comment}
\usepackage{amsthm}
\usepackage{amsmath}

\usepackage{bm}

\usepackage{graphicx}
\usepackage{xspace}
\usepackage{stmaryrd}
\usepackage{subcaption}
\usepackage{array}
\usepackage{hyperref}

\newcommand{\descr}[1]{\vspace{0.1cm}\noindent\textit{#1}}
\usepackage[utf8]{inputenc}
\newtheorem{definition}{Definition}
\newtheorem{theorem}{Theorem}
\newtheorem{lemma}{Lemma}
\usepackage{stmaryrd}
\usepackage[ruled,noline, linesnumbered]{algorithm2e} 
\usepackage{slashbox}
\usepackage{makecell}

\usepackage{resizegather}

\SetAlgoNoLine
\SetAlCapNameFnt{\small}
\SetAlCapFnt{\small}

\def\BibTeX{{\rm B\kern-.05em{\sc i\kern-.025em b}\kern-.08emT\kern-.1667em\lower.7ex\hbox{E}\kern-.125emX}}
    
\newcommand{\mbf}[1]{{\mathbf{#1}}}
\newcommand{\Adv}{\mathsf{Adv}}

\newcommand{\Tcl}{T_{\mathsf{cl}}}
\newcommand{\Tgd}{T_{\mathsf{gd}}}

\newcommand{\E}[1]{\mathrm{E}\left[ #1 \right]}
\newcommand{\Var}[1]{\mathrm{Var}\left[ #1 \right]}
\newcommand{\Prob}[1]{\mathrm{P}\left[ #1 \right]}

\usepackage{enumitem}
\usepackage{lipsum}

\setlist[itemize]{leftmargin=*}

\begin{document}

\title{Federated Learning in Adversarial Settings}

\author{Raouf Kerkouche}
\email{raouf.kerkouche@inria.fr}
\affiliation{%
  \institution{Privatics team, Univ. Grenoble Alpes, Inria, 38000 Grenoble, France}
}

\author{Gergely \'Acs}
\email{acs@crysys.hu}
\affiliation{%
  \institution{Crysys Lab, BME-HIT}
}

\author{Claude Castelluccia}
\email{claude.castelluccia@inria.fr}
\affiliation{%
  \institution{Privatics team, Univ. Grenoble Alpes, Inria,
 38000 Grenoble, France}
}


\begin{abstract}
  Federated Learning enables entities to collaboratively learn a shared prediction model while keeping their training data locally. It prevents data collection and aggregation and, therefore, mitigates the associated privacy risks. However, it still remains vulnerable to various security attacks where malicious participants aim at degrading the generated model, inserting backdoors, or inferring other participants' training data. This paper presents a new federated learning scheme that provides different trade-offs between robustness, privacy, bandwidth efficiency, and model accuracy. Our scheme uses biased quantization of model updates and hence is bandwidth efficient. It is also robust against state-of-the-art backdoor as well as model degradation attacks even when a large proportion of the participant nodes are malicious. We propose a practical differentially private extension of this scheme which protects the whole dataset of participating entities. We show that this extension performs as efficiently as the non-private but robust scheme, even with stringent privacy requirements but are less robust against model degradation and backdoor attacks. This suggests a possible fundamental trade-off between Differential Privacy and robustness. 
\end{abstract}



\keywords{Differential Privacy, Privacy-preserving, Security, Federated learning, robustness, bandwidth efficient.}


\maketitle

\pagestyle{plain} 

\section{Introduction}
In standard centralized training, a machine learning model is generated by a single server who collects the training data from different sources such as mobile devices, sensors, or organizations. 
However, data owners are often reluctant to share their potentially sensitive data with an untrusted server.
To overcome this shortcoming, \emph{Collaborative Learning}  allows several parties (clients)  to build a common model  without sharing their private training data.   
It proposes to distribute and run the Machine Learning algorithms on the entities that own the data instead of a central server. Data owners periodically synchronize their local models either distributively or through a central, perhaps untrusted server. For example, in Federated Learning, clients send their model updates to the central server which then summarizes the weights into a common model and returns this model to the clients for another round. This protocol repeats until the model converges. Federated Learning has been gaining popularity and considered to train shared models for many applications such as input text prediction, ad selection\footnote{\url{https://blog.chromium.org/2019/08/potential-uses-for-privacy-sandbox.html}}, drug discovery\footnote{\url{https://www.melloddy.eu}}, or various medical applications \cite{choudhury2019differential} over the confidential data of many different entities.


Although Collaborative learning creates new opportunities, it also has a few drawbacks. First, it is \emph{not robust} against misbehaving parties who may not follow the learning protocol faithfully in order to degrade the model performance. For example, a malicious party may send bogus model updates for aggregation which degrades the overall model quality or introduces backdoors \cite{IBM_attack}. Second, malicious parties can potentially \emph{extract private information} about the training data of honest parties from their model updates  or the common model \cite{Property_inference,NasrSH19}. Third, the \emph{bandwidth requirement} of Federated Learning can be significant for large models: each update is typically composed of $32 \cdot n$ bits, where $n$ is the number of model parameters. Since it is not unusual to have models with  thousands or even millions of parameters, the size of each update can be quite large.





The motivation of this paper is to propose a new Federated Learning scheme that mitigates these issues, and is (1) \emph{bandwidth efficient}, (2) \emph{privacy-preserving}, (3) \emph{secure}, and (4) \emph{accurate}. 
We identify several trade-offs between these design goals suggesting that satisfying all these requirements simultaneously is inherently difficult. 
More
specifically, we make the following contributions:
\begin{itemize}
    \item[-] We adapt the signSGD learning algorithm \cite{SIGNSGD, SIGNSGD_vote_robustness} to the federated learning setting which sends only a single bit per model parameter for aggregation instead of their actual value. This extreme quantization reduces the required updates' bandwidth by a factor of 32, while still providing similar performance to the standard centralized federated learning approach. As only a small random subset of participants send their updates at each federated round, our proposal is not only more bandwith efficient but also provides stronger robustness and privacy guarantees than the standard signSGD algorithm (see Section \ref{sec:related_work} for a detailed comparison).    
 
    \item[-] We propose a privacy-preserving federated signSGD extension which provides client-level Differential Privacy (DP).  Specifically, it hides any information that is unique to a client's training data, regardless whether it is about a single or multiple records, but still allows learning about characteristics that are common among multiple clients' training data.
    We show that our DP learning protocol, whose convergence rate is also computed analytically, produces models with an accuracy comparable to the non-private federated case, even with stringent privacy guarantees (e.g., $\varepsilon=1$). 
    
    \item[-] In order to diminish the communication costs of our DP algorithm, we propose a novel discretized and distributed version of the Gaussian Mechanism. In particular, as opposed to the standard Gaussian Mechanism \cite{Dwork2014book}, the noise values come from a discretized domain and are tightly concentrated around its mean depending on the desired privacy guarantee $\varepsilon$.  
    As a result, these values can be encoded with fewer bits than if they came from a continuous Gaussian distribution.   

    \item[-] We experimentally evaluate the robustness of our schemes by implementing and testing several State-of-the-Art security attacks such as model degradation (where the adversary aims at modifying the global model) or backdoor inclusion attacks (where the adversary aims at inserting hidden backdoors). We show that, due to the quantization of model updates, the non-private federated signSGD protocol, called SignFed, is more resilient  to these attacks than the standard Federated Learning scheme. However, its differentially private variant turns out to be more vulnerable to the security attacks. Indeed, the attacks are inherently concealed by the noise which is introduced to guarantee Differential Privacy.
    
\end{itemize}

\textbf{Organization:} The paper is structured as follows. Section \ref{sec:backg} details the preliminaries including the basic federated learning algorithm (StdFed) and Differential Privacy (DP). Section \ref{sec:fl_sign} describes SignFed which is an adaption of signSGD \cite{SIGNSGD_vote_robustness} to the federated learning setting. The performance of SignFed and StdFed are compared in Section \ref{sec:fl_sign_per}. A DP learning protocol for client-level privacy (DP-SignFed) is presented in Section \ref{sec:priv_fl_sign}, whose performance are evaluated in Section \ref{sec:eval}. The resistance of our protocols against various security attacks is studied in Section \ref{sec:security}. Section \ref{sec:related_work} compares our proposal with prior work, and finally Section \ref{sec:discussion} provides a summary and additional discussions about the proposed algorithms.

\section{Background}
\label{sec:backg}
\subsection{Federated Learning (StdFed)}
\label{StdFed}
\begin{algorithm}[!t]
\small
		\caption{StdFed: Federated Learning \label{alg:fed_learn}}
	\DontPrintSemicolon
	{\bf Server:}\;
	\Indp Initialize common model $w_0$\;
	\For {$t=1$ \KwTo $\Tcl$}
	{
	    Select $\mathbb{K}$ clients uniformly at random \;
		\For {\textrm{each} client $k$ \textrm{in} $\mathbb{K}$}
		{	
			$\Delta \mbf{w}_t^k = \mathbf{Client}_k(\mbf{w}_{t-1})$\;
		}
		$\mbf{w}_{t} = \mbf{w}_{t-1} + \sum_{k} \frac{|D_k|}{\sum_j |D_j|} \Delta \mbf{w}_{t}^{k}$\;
	}
	\KwOut{Global model $\mbf{w}_t$}\;
	\Indm {\bf $\mathbf{Client}_{k}(\mbf{w}_{t-1}^k)$:}\;
	\Indp
	$\mbf{w}_{t}^k = \mathbf{SGD}(D_k, \mbf{w}_{t-1}^k, \Tgd)$\;
	\KwOut{Model update $(\mbf{w}_{t}^k- \mbf{w}_{t-1}^k)$} 
\end{algorithm}

\begin{algorithm}[h]
\small
	\caption{Stochastic Gradient Descent \label{alg:sgd}}
	\DontPrintSemicolon
	\KwIn{$D$ : training data, $\Tgd$ : local epochs, $\mathbf{w}$ : weights}  
    \For {$t=1$ \KwTo $\Tgd$}
	{
	    Select batch $\mathbb{B}$ from $D$ randomly\;
	    $\mbf{w} = \mbf{w} - \eta \nabla f(\mathbb{B}; \mbf{w})$\;
	}
    \KwOut{Model $\mbf{w}$} 
\end{algorithm}

In federated learning \cite{ShokriS15,FedAVG}, multiple parties (clients) build a common machine learning model on the union of their training data without sharing them with each other. At each round of the training, some clients retrieve the global model from the parameter server, update the global model based on their own training data, and send back their updated model to the server. The server aggregates the updated models of all clients to obtain a global model that is re-distributed to some selected parties in the next round.  

In particular, a subset $\mathbb{K}$ of all $N$ clients are randomly selected at each round to update the global model, and $C = |\mathbb{K}| / N$ denotes the fraction of selected clients. At round $t$, a selected client $k \in \mathbb{K}$ executes $\Tgd$ local gradient descent iterations on the common model $\mbf{w}_{t-1}$ using its own training data $D_k$ ($D = \cup_{k\in \mathbb{K}} D_k$), and obtains the updated model $\mbf{w}_{t}^k$, where the number of weights is denoted by $n$ (i.e., $|\mbf{w}_t^{k}| = |\Delta \mbf{w}_t^k| = n$ for all $k$ and $t$). Each client $k$ submits the update $\Delta \mbf{w}_{t}^k = \mbf{w}_{t}^k - \mbf{w}_{t-1}^k$ to the server, which then updates the common model as follows: $\mbf{w}_{t} = \mbf{w}_{t-1} + \sum_{k \in \mathbb{K}} \frac{|D_k|}{\sum_j |D_j|} \Delta \mbf{w}_{t}^k$, where $|D_k|$ is known to the server for all $k$ (a client's update is weighted with the size of its training data). 
The server stops training after a fixed number of rounds $\Tcl$, or when the performance of the common model does not improve on a held-out data. 

Note that each $D_k$ may be generated from different distributions (i.e., Non-IID case), that is, any client's local dataset may not be representative of the population distribution \cite{FedAVG}. This can happen, for example, when not all output classes are represented in every client's training data. 
The federated learning of neural networks is summarized in Alg.~\ref{alg:fed_learn}. In the sequel, each client is assumed to use the same model architecture.


The motivation of federated learning is three-fold: first, it aims to provide confidentiality of each participant's training data by sharing only model updates instead of potentially sensitive training data. Second, in order to decrease communication costs, clients can perform multiple local SGD iterations before sending their update back to the server. 
Third, in each round, only a few clients are required to perform local training of the common model, which not only further diminishes communications costs  but also increases robustness against temporary client failures and hence makes the approach especially appealing with large number of clients.

However, several prior works have demonstrated that model updates do leak potentially sensitive information \cite{NasrSH19,Property_inference}. Hence, simply not sharing training data \emph{per se} is not enough to guarantee their confidentiality.

\begin{algorithm}[!t]
\small
		\caption{SignFed: Sign Federated Learning  \label{alg:sign_fed_learn}}
	\DontPrintSemicolon
	{\bf Server:}\;
	\Indp Initialize common model $w_0$\;
	\For {$t=1$ \KwTo $\Tcl$}
	{
	    Select $\mathbb{K}$ clients uniformly at random \;
		\For {each client $k$ \textrm{in} $\mathbb{K}$}
		{	
			$\mbf{s}_t^k = \mathbf{Client}_k(\mbf{w}_{t-1})$\;
		}
		$\mbf{w}_{t} = \mbf{w}_{t-1} + \gamma \mathsf{sign}\left(\sum_{k} \mbf{s}_{t}^{k}\right)$\;
	}
	\KwOut{Global model $\mbf{w}_t$}\;
	\Indm {\bf $\mathbf{Client}_{k}(\mbf{w}_{t-1}^i)$:}\;
	\Indp
	$\mbf{w}_{t}^k= \mathbf{SGD}(D_k, \mbf{w}_{t-1}^k, \Tgd)$\;
    \KwOut{Model update $\mathsf{sign}(\mbf{w}_{t}^k - \mbf{w}_{t-1}^k)$} 
\end{algorithm}

\subsection{Differential Privacy}
\label{sec:DP}
Differential privacy allows a party to privately release information about a dataset:  a function of an input dataset is perturbed, so that any information which can differentiate a record from the rest of the dataset is bounded~\cite{Dwork2014book}.

 \begin{definition}[Privacy loss]
 Let $\mathcal{A}$ be a privacy mechanism which assigns a value $\mathit{Range}(\mathcal{A})$ to a dataset $D$. The privacy loss of $\mathcal{A}$ with datasets $D$ and $D'$ at output $O \in \mathit{Range}(\mathcal{A})$ is a random variable $\mathcal{P}(\mathcal{A},D,D',O) = \log\frac{\Pr[\mathcal{A}(D) = O]}{\Pr[\mathcal{A}(D') = O]}$ 
 where the probability is taken on the randomness of $\mathcal{A}$.
 \label{def:ploss}
 \end{definition}

\begin{definition}[$(\epsilon,\delta)$-Differential Privacy~\cite{Dwork2014book}] 
A privacy mechanism $\mathcal{A}$ guarantees $(\varepsilon, \delta)$-differential privacy if for any database $D$ and $D'$, differing on at most one record, $\Pr_{O \sim \mathcal{A}(D)}[\mathcal{P}(\mathcal{A},D,D',O) > \varepsilon] \leq \delta$. 

\label{def:DP}
\end{definition}

Intuitively, this guarantees that an adversary, provided with the output of $\mathcal{A}$, can draw almost the same conclusions (up to $\varepsilon$ with probability larger than $1 - \delta$) about any record no matter if it is included in the input of $\mathcal{A}$ or not~\cite{Dwork2014book}. That is, for any record owner, a privacy breach is unlikely to be due to its participation in the dataset. 

\descr{Moments Accountant.} Differential privacy maintains composition; the privacy guarantee of the  $k$-fold adaptive composition  of $\mathcal{A}_{1:k} = \mathcal{A}_1, \ldots, \mathcal{A}_k$ can be computed using the moments accountant method \cite{Abadi}. In particular, it follows from Markov's inequality that $\Pr[\mathcal{P}(\mathcal{A},D,D',O) \geq \varepsilon] \leq \mathbb{E}[\exp(\lambda \mathcal{P}(\mathcal{A},D,D',O))]/\exp(\lambda\varepsilon)$ for any output $O \in \mathit{Range}(\mathcal{A})$ and $\lambda > 0$. This implies that $\mathcal{A}$ is $(\varepsilon, \delta)$-DP 
with $\delta = \min_{\lambda} \exp(\alpha_{\mathcal{A}}(\lambda) - \lambda \varepsilon)$, where $\alpha_{\mathcal{A}}(\lambda) = \max_{D,D'} \log\mathbb{E}_{O\sim \mathcal{A}(D)}[\exp(\lambda \mathcal{P}(\mathcal{A},D,D',O))]$ is the log of the moment generating function of the privacy loss. The privacy guarantee of the composite mechanism $\mathcal{A}_{1:k}$ can be computed using that $\alpha_{\mathcal{A}_{1:k}}(\lambda) \leq \sum_{i=1}^k \alpha_{\mathcal{A}_{i}}(\lambda)$ \cite{Abadi}. \smallskip 

\descr{Gaussian Mechanism.} There are a few ways to achieve DP, including the Gaussian mechanism~\cite{Dwork2014book}. A fundamental concept of all of them is the \emph{global sensitivity} of a function~\cite{Dwork2014book}.
\begin{definition}[Global $L_p$-sensitivity] 
For any function $f:\mathcal{D} \rightarrow \mathbb{R}^ n$, the $L_p$-sensitivity of $f$ is
$\Delta_p f = \max_{D, D'} || f(D)-f(D') ||_p$, 
for all $D, D'$ differing in at most one record, where $||\cdot||_p$ denotes the $L_p$-norm.\vspace*{-0.15cm}
\label{def:global_sens}
\end{definition}
\smallskip
The Gaussian Mechanism~\cite{Dwork2014book} 
consists of adding Gaussian noise to the true output of a function.
In particular, for any function $f:\mathcal{D} \rightarrow \mathbb{R}^n$, the Gaussian mechanism is defined as adding i.i.d Gaussian noise with variance $(\Delta_2 f \cdot \sigma)^2$  and zero mean to each coordinate value of  $f(D)$. Recall that the pdf of the Gaussian distribution with mean $\mu$ and variance $\xi^2$ is
\begin{align}
\label{eq:g_pdf}
\mathsf{pdf}_{\mathcal{G}(\mu, \xi)}(x) = \frac{1}{\sqrt{2\pi}\xi} \exp\left(-\frac{(x-\mu)^2}{2 \xi^2}\right) 
\end{align}

In fact, the Gaussian mechanism draws vector values from a multivariate spherical (or isotropic) Gaussian distribution
which is described by random variable $\mathcal{G}(f(D), \Delta_2 f \cdot \sigma\mathbf{I}_n)$, where $n$ is omitted if its unambiguous in the given context.

\section{SignFed: SIGN Protocol in the Federated Learning Setting}
\label{sec:fl_sign}


\subsection{The SignFed Protocol}

In the StdFed scheme, presented in Section~\ref{StdFed}, each selected client sends its updated model to the central server. As discussed previously, this scheme has several drawbacks in terms of bandwidth, robustness and privacy. We propose to limit these drawbacks by quantizing the model weights as in \cite{SIGNSGD_vote_robustness}. More specifically, in the new scheme, referred to as SignFed in the rest of this paper, each client sends only the sign of every coordinate value in its parameter update vector. 
The server takes the sign of the sum of signs per coordinate and scales down the result with a fixed constant $\gamma$ (which is in the order of $10^{-3}$ in practice) in order to limit the contribution of each client and adjust convergence. This scaled aggregated updates are added to the global model.

More specifically, SignFed (see Alg.~\ref{alg:sign_fed_learn}) differs from the standard federated scheme StdFed (see  Alg.~\ref{alg:fed_learn}) as follows:
\begin{enumerate}
\item Each client returns $\mbf{s}_t^k = \mathsf{sign}(\mbf{w} - \mbf{w}_{t-1}^k)$ instead of $(\mbf{w} - \mbf{w}_{t-1}^k)$, where $\mathsf{sign} : \mathbb{R}^n \rightarrow \{ -1,1\}^n$ returns the sign of each coordinate value of the input vector if it is non-zero and a sign chosen uniformly at random otherwise.

\item The server sums the sign vectors $\mbf{s}_t^k$ sent by each client $k$ and computes the sign vector of this sum as $\mathsf{sign}\left(\sum_{k}\mbf{s}_t^k\right)$. This is equivalent to take the median of all clients' signs at every position of the update vectors. Unlike in Alg.~\ref{alg:fed_learn}, the update $\mbf{s}_t^k$ is \emph{not} weighted with client $k$'s data size $|D_k|$, since that would require the client to send $|D_k|$ to the server which would enable the adversary to maliciously scale up its sign vector by sending a fabricated size of its training data.
\end{enumerate}

The extreme quantization performed by SignFed reduces the communication costs of federated learning by a factor of 32 (since only one bit is sent per parameter instead of 32 bits), and also, as we will demonstrate later, improves its robustness against different attacks aiming to maliciously manipulate the common model through the updates. Note also that, if the quantized update vector is sparse, other compression techniques can further improve communication efficiency \cite{Quant_Fed}. 


\subsection{Experimental Set-up}

This section describes the experimental set-up that are used 
to evaluate the accuracy, security and privacy of our proposals in the rest of the paper. The following datasets were used: MNIST, Fashion-MNIST, IMDB, LFW and CIFAR which is augmented from 50,000 images to 500,000 (See Appendix~\ref{sec:datasets} for more details)

\subsection{Performance evaluation}
\label{sec:fl_sign_per}
In this section, we compare the performance of SignFed and StdFed using the same configuration;
Table~\ref{tab:fl_sign_param} summarizes the different parameter values that were
used for the different datasets.
For SignFed, $\gamma $, the learning rate, was set to  $0.001$ for all datasets\footnote{We noticed experimentally that $\gamma$ should be selected between range 0.001 and 0.005. And it should be increased when DP is used.}. $N$, the total number of participant clients, was set to $1000$. $C$, the the percentage of selected clients at each round, was set to $0.1$. $|D_k|$ is the training data size of client $k$. $|\mathbb{B}|$, the batch size, was set to 50 with CIFAR dataset, 25 with IMDB, 10 for MNIST and Fashion-MNIST datasets. $\Tgd$, the local gradient descent iterations per round and per client, was set to $30$, $30$, $5$ and $50$ for MNIST, Fashion-MNIST, IMDB and CIFAR, respectively.
$\Tcl$, the number of rounds, was set to $100$ for the MNIST, Fashion-MNIST, IMDB datasets, and $400$ for CIFAR.
 We use two optimizers: the stochastic gradient descent (SGD) \cite{KERAS_optimizers} with a learning rate ($\eta$) set to $0.215$ and the adaptive moment estimation (Adam) \cite{Adam} \cite{KERAS_optimizers} with a learning rate set to $0.001$. 

The global model accuracy of StdFed and SignFed on
the CIFAR, MNIST, Fashion-MNIST, IMDB datasets are compared in Table \ref{tab:fl_sign_evaluation}. The results show that the accuracy performance of both schemes over the four datasets are very similar despite the severe parameter quantization.

The bandwidth consumption is calculated by measuring the average number of bits sent by a client to the server. This is computed as ($C  \times$ best\_round  $\times$  model\_size) for SignFed, and ($32 \times C \times $ best\_round $ \times $ model\_size) for StdFed, where model\_size is the number of the model parameters and best\_round represents the round when we get the best accuracy over $\Tcl$ rounds.

\begin{table*}[!h]
	\caption{Model accuracy and average bandwidth consumption from a client to the server (Megabytes) based on the best round (over $\Tcl$ rounds) in terms of accuracy.}
	\label{tab:fl_sign_evaluation}
	\centering
	\begin{tabular}{c|c|c|c|c|c|c}
		\hline
		  \multirow{2}{*}{Dataset} & \multicolumn{3}{|c|}{StdFed}  & \multicolumn{3}{|c}{SignFed}  \\
		\cline{2-7}
		
		&  Acc & round & Cost &  Acc & round & Cost  \\
		\cline{2-7}
		CIFAR	            & 	0.86 & 375 & 205.46	    & 0.83 & 386 & 6.61 \\
		MNIST      	        &   0.99 & 88 & 58.55 	   &  0.98 & 48 & 1.0     \\
		Fashion-MNIST      	&   0.89 & 90 &  59.88	   &  0.87 & 68 &  1.41 \\
		IMDB      	        &   0.88 & 84 & 13.53 	   &  0.85  & 91 & 0.46     \\
		\hline
	\end{tabular}
\end{table*}


\section{Privacy-Preserving SignFed}
\label{sec:priv_fl_sign}
In SignFed, a participant only sends the signs of its updates, as opposed to their actual value, hence it intuitively reveals less information 
about the client's dataset than the original StdFed scheme. In order to experimentally validate this intuition, we implemented the inference 
attack described in \cite{Property_inference} on StdFed and SignFed\footnote{A model is trained for gender classification on the LFW dataset. The adversary's goal is to infer from the model updates whether a specific group of individuals in a client's dataset are black.}. The results, which are not reported in this paper for lack of space, clearly validated our intuition (the attack accuracy dropped from 92\% for StdFed to 50\% for SignFed). Although these results are very promising and might confirm 
that privacy is preserved in practice, it does not provide any strong guarantees. In order to obtain theoretically private schemes, we 
extend SignFed with Differential Privacy. Our goal is to design differentially private schemes that are efficient in terms of accuracy and 
bandwidth (even for small $\varepsilon$ values).

\subsection{Privacy Model}
\label{sec:privacy_model}
We consider an adversary, or a set of colluding adversaries, who can access any update vector sent by the server or any clients at each round of the protocol.  A plausible adversary is a participating entity, i.e. a malicious client or server, that wants to infer the training data used by other participants.
The adversary is \emph{passive} (i.e., honest-but-curious), that is, it follows the learning protocol faithfully. 

Different privacy requirements can be considered depending on what information the adversary aims to infer. In general, private information can be inferred about:
\begin{itemize}
    \item any record (user) in any dataset of any client (\emph{record-level privacy}),
    \item any client/party (\emph{client-level privacy}).
\end{itemize}

To illustrate the above requirements, suppose that several banks build a common model to predict the creditworthiness of their customers. A bank certainly does not want other banks to learn the financial status of any of their customers (record privacy) and perhaps not even the average income of all their customers  (client privacy).

Record-level privacy is a standard requirement used in the privacy literature and is usually weaker than client-level privacy. Indeed, client-level privacy requires to hide any information which is unique to a client including perhaps all its training data.   


We aim at developing a solution that provides \emph{client-level privacy and is also bandwidth efficient}. 
For example, in the scenario of collaborating banks, we aim at protecting any information that is unique to each single bank's training data.
The adversary should not be able to learn from the received model or its updates whether any client's data is involved in the federated run (up to $\varepsilon$ and $\delta$). 
We believe that this adversarial model is reasonable in many practical applications when the confidential information spans over multiple samples in the training data of a single client (e.g., the presence of a group a samples, such as people from a certain race). Differential Privacy guarantees plausible deniability not only to any groups of samples of a client but also to any client in the federated run. Therefore, any negative privacy impact on a party (or its training samples) cannot be attributed to their involvement in the protocol run.

\subsection{Client-Based Privacy Preserving Federated Learning (DP-SignFed)}
\label{sec:fl_client_dp}
To guarantee differential privacy per client, every client should add enough noise to its update locally such that the server cannot learn any client-specific information from the noisy update. However, this approach (aka, local differential privacy \cite{ErlingssonPK14}) requires so much perturbation that it is impractical if the number of clients is limited. Instead, likewise \cite{Hybrid-approach}, we follow a different approach where clients themselves add noise in a distributed manner so that the aggregated updates are sufficiently noised to have meaningful differential privacy.
To this end, individual noisy updates are encrypted with a simple and efficient encryption scheme taken from \cite{AcsC11,BonawitzIKMMPRS16}. The purpose of this encryption is to prevent the adversary from accessing the individual (and weakly-noised) update per client but only their sum over all clients which is in turn sufficiently noised to guarantee DP for any client. 


Specifically, each client $k$ first computes the gradient update $\Delta{\mbf{w}}_t^k$ (in Line 12 of Alg.~\ref{alg:sign_client_fed_learn}) and then takes the sign vector of this update. Then,  
a random noise share $\rho_k$ is added to the sign vector $\mathsf{sign}(\Delta \mbf{w}_t^k)$ so that $\sum_{k \in \mathbb{K}} \mathsf{sign}(\Delta \mbf{w}_t^k)  + \rho_k$ satisfies differential privacy. A simple solution is that $\rho_i \sim \mathcal{G}(0, \sqrt{n} \sigma \mathbf{I} / \sqrt{|\mathbb{K}|})$, which means that  $\sum_{k \in \mathbb{K}} \mathsf{sign}( \Delta{\mbf{w}}_t^k) + \sum_{k \in \mathbb{K}} \rho_k = \sum_{k \in \mathbb{K}} \mathsf{sign}(\Delta{\mbf{w}}_t^k) + \mathcal{G}(0, \sqrt{n} \mathbf{I} \sigma)$ as the sum of Gaussian random variables also follows Gaussian distribution\footnote{More precisely, $\sum_i \mathcal{G}(\nu_i, \xi_i) = \mathcal{G}\left(\sum_i \nu_i, \sqrt{\sum_i \xi_i^2}\right)$}. Indeed, the variance of the Gaussian noise has to be proportional to the $L_2$-sensitivity of the sign vector which is no more than $\sqrt{n}$, where $n$ is the number of parameters.

However, recall that the adversary can access $\mathsf{sign}(\Delta{\mbf{w}}_t^k) + \rho_k$, which means that, if $|\mathbb{K}|$ is too large, $\rho_k$ is likely to be small allowing the adversary to learn $\mathsf{sign}(\Delta{\mbf{w}}_t^k) + \rho_k$ very accurately. For this reason, each client $k$ encrypts $\mathsf{sign}(\Delta{\mbf{w}}_t^k) + \rho_k$ and sends the encrypted result  to the aggregator. After summing all
the encrypted values, the server obtains
$
\sum_k \mathsf{Enc}_{\mathbf{K}_k}(\mathsf{sign}(\Delta{\mbf{w}}_t^k) + \rho_k) = \sum_{k} \mathsf{sign}(\Delta{\mbf{w}}_t^k) + \mathcal{G}\left(0,  \sqrt{n}\sigma\mathbf{I}\right) 
$
where $\mathsf{Enc}_{\mathbf{K}_k}(\mathsf{sign}(\Delta{\mbf{w}}_t^k) + \rho_k) = \mathsf{sign}(\Delta{\mbf{w}}_t^k) + \rho_k + \mathbf{K}_k \mod m$ and  $\sum_k \mathbf{K}_k = \mathbf{0}$ (see \cite{AcsC11, BonawitzIKMMPRS16} for details). Here, modulo is taken element-wise and $m= 2^{\lceil \log_2(\max_k||1 + \rho_k||_{\infty} |\mathbb{K}|)\rceil}$. Therefore, the server can only access the aggregate which is sufficiently noised to guarantee differential privacy; any client-specific information that could be learnt from the noisy aggregate is quantified by the moments accountant described in Section \ref{sec:DP}. 
To make learning more resilient to perturbation, the server takes the sign of the sum of updates and scales the result with $\gamma < 1$ which is crucial to achieve convergence in practice especially if $\sqrt{n}\sigma$ is large. 

Unfortunately, the above simple approach is not bandwidth efficient; adding noise from the continuous domain requires each noisy update $\mathsf{sign}(\Delta{\mbf{w}}_t^k) + \rho_k$ to be encoded as a floating-point number\footnote{and then as a large integer for encryption} (represented by at least 32 bits on a commodity hardware) no matter that $\mathsf{sign}(\Delta{\mbf{w}}_t^k)$ would need only 1 bit per coordinate. Therefore, the noisy update needs at least 32 times more data to be transferred from a client to the server than with SignFed (in Alg.~\ref{alg:sign_fed_learn}). 

To alleviate the above bandwidth problem, each client $k$ generates a random integer from a \emph{discrete} Gaussian distribution with mean $\mathsf{sign}(\Delta{\mbf{w}}_t^k)$, encrypts this random integer, and sends the result for aggregation. Since the discrete Gaussian random variable has an integer value and is concentrated around its mean, its value can be encoded with fewer bits than a floating-point number. The new learning algorithm, called  DP-SignFed, guarantees differential privacy for any client and is summarized in Alg. \ref{alg:sign_client_fed_learn}.  

In what follows, we first describe the Discrete Gaussian Mechanism (DGM), which is used in DP-SignFed, and prove that it practically provides the same privacy guarantee as the continuous Gaussian Mechanism (GM) if its variance is sufficiently large. This allows us to precisely quantify the privacy guarantee of DP-SignFed. Finally, we show that using DGM instead of (continuous) GM in DP-SignFed reduces the communication overhead by roughly 40\%. 

\subsubsection{Discrete Gaussian Mechanism (DGM)}

The discrete Gaussian distribution has probability mass function 
\begin{align}
\label{eq:g_pmf}
    \mathsf{pmf}_{\mathcal{DG}(\mu, \xi)}(x) = Z^{-1}\exp(-(x-\mu)^2/2\xi^2)
\end{align} where $Z = \sum_{x \in \mathbb{Z}} \exp(- (x-\mu)^2/2\xi^2)$. Note that 
$\mu \in \mathbb{R}$ but the support of $\mathcal{DG}$ is always $\mathbb{Z}$. 
Although $Z$ is infeasible to compute, there are several efficient techniques to sample from the discrete Gaussian distribution \cite{Micciancio017}.



The next lemma shows that the pmf of the discrete Gaussian distribution can be almost perfectly approximated by its continuous counterpart if $\xi$ is large enough.

\begin{lemma}
\label{lem:cont_app}
Let $\mathsf{pmf}_{\mathcal{DG}(\mu, \xi)}(x)$ and $\mathsf{pdf}_{\mathcal{G}(\mu, \xi)}(x)$ be as defined in Eq. \eqref{eq:g_pmf} and Eq. \eqref{eq:g_pdf}, respectively, and $\kappa(\xi) =\frac{2e^{-2\pi^2\xi^2}}{1-e^{-6\pi^2\xi^2}}$.
Then,
$
  1 - \kappa(\xi)\leq \frac{\mathsf{pdf}_{\mathcal{G}(\mu, \xi)}(x)}{\mathsf{pmf}_{\mathcal{DG}(\mu, \xi)}(x)} \leq  1 + \kappa(\xi)
$
for $x \in \mathbb{Z}$.
\end{lemma}
The proof can be found in Appendix \ref{app:proof_lem_cont_app}.

The multivariate spherical version of $\mathcal{DG}$ can be defined analogously to  the spherical Gaussian distribution, that is, if $\mathbf{z} \sim \mathcal{DG}(\bm{\mu}, \xi)$, then $z_i \sim \mathcal{DG}(\mu_i, \xi)$ independently for each $i$.  

The Discrete Gaussian Mechanism (DGM) is defined analogously to the (continuous) Gaussian Mechanism except that it uses discrete Gaussian noise instead of its continuous counterpart for perturbation.
The next theorem shows that the moments of DGM can be tightly upper bounded by that of the continuous Gaussian mechanism if $\xi$ is large enough, and hence the privacy guarantee of DGM can be efficiently and accurately approximated.

Let $\eta_0^{\mathcal{G}}(x|\xi) =  \mathsf{pdf}_{\mathcal{G}(0, \xi)}(x)$ and $\eta_1^{\mathcal{G}}(x|\xi) =  (1-C) \mathsf{pdf}_{\mathcal{G}(0, \xi)}(x) + C \mathsf{pdf}_{\mathcal{G}(1, \xi)}(x)$ where $C$ is the sampling probability of a single client in a single round. Let
\begin{align}
\label{eq:alpha}
\alpha_{\mathcal{G}}(\lambda| C) &= \log\max(E_1(\lambda, \xi, C), E_2(\lambda, \xi, C)) 
\end{align}
where
$
E_1(\lambda,  \xi, C) =  \int_{\mathbb{R}}\eta_0^{\mathcal{G}}(x|\xi, C) \cdot \left(\frac{\eta_0^{\mathcal{G}}(x|\xi, C)}{\eta_1^{\mathcal{G}}(x|\xi, C)}\right)^{\lambda} dx
$ and
$ E_2(\lambda,  \xi, C) = \int_{\mathbb{R}}\eta_1^{\mathcal{G}}(x|\xi, C) \cdot \left(\frac{\eta_1^{\mathcal{G}}(x|\xi, C)}{\eta_0^{\mathcal{G}}(x|\xi, C)}\right)^{\lambda} dx
$.
$\alpha_{\mathcal{DG}}(\lambda|C)$ is defined analogously to $\alpha_{\mathcal{G}}(\lambda|C)$. 

\begin{theorem}[Privacy of DGM]
\label{thm:dg_privacy}
$\alpha_{\mathcal{DG}}(\lambda | C) \leq \alpha_\mathcal{G}(\lambda | C) + 
\log \left(\frac{(1 + \kappa(\xi))^{\lambda}}{(1 - \kappa(\xi))^{\lambda+1}}\right)$ for any $C$, where  $\kappa(\xi)$ is defined in Lemma \ref{lem:cont_app}. Therefore, DGM is $(\min_\lambda  \left(\alpha_\mathcal{G}(\lambda | C) + 
\log \left(\frac{(1 + \kappa(\xi))^{\lambda}}{(1 - \kappa(\xi))^{\lambda+1}}\right)\right) - \log \delta) /\lambda, \delta)$-DP. 
\end{theorem}
The proof can be found in Appendix \ref{app:proof_thm_dg_privacy}. Given a fixed value of $\delta$, $\varepsilon$ is computed numerically  as in \cite{Abadi,MironovTZ19}.

\subsubsection{Privacy of DP-SignFed}

As shown in Alg.~\ref{alg:sign_client_fed_learn}, each client $k$ generates a random integer vector $\mathbf{z}_k \sim \mathcal{DG}(\mathsf{sign}(\Delta{\mbf{w}}_t^k), \sqrt{n}\sigma\mathbf{I}/\sqrt{|\mathbb{K}|})$ in DP-SignFed. Then, every client sends the encrypted result $\mathsf{Enc}_{\mathbf{K}_k}(\mathbf{z}_k)$ to the aggregator. After summing all
the encrypted integers, the server obtains
\begin{align}
\label{eq:enc}
\sum_k \mathsf{Enc}_{\mathbf{K}_k}(\mathbf{z}_k) &= \sum_k \mathbf{z}_k = \sum_k \mathcal{DG}(\mathsf{sign}(\Delta{\mbf{w}}_t^k),  \sqrt{n}\sigma\mathbf{I}/\sqrt{|\mathbb{K}|}) 
\end{align}

The next theorem, proved in Appendix \ref{app:proof_thm_fl_client}, shows that DP-SignFed is differentially private, supposing that the adversary can only access $\sum_k \mathcal{DG}(\mathsf{sign}(\Delta{\mbf{w}}_t^k),  \sqrt{n}\sigma\mathbf{I}/\sqrt{|\mathbb{K}|})$ except any of its members $\mathcal{DG}(\mathsf{sign}(\Delta{\mbf{w}}_t^k),  \sqrt{n}\sigma\mathbf{I}/\sqrt{|\mathbb{K}|})$. 

\begin{theorem}[Privacy of DP-SignFed]
\label{thm:fl_client}
For any $\delta >0$, DP-SignFed is $(\min_\lambda (T \cdot \left(\alpha_\mathcal{G}(\lambda | C) + 
\log \left(\frac{(1 + \kappa(\sqrt{n}\sigma))^{\lambda}}{(1 - \kappa(\sqrt{n}\sigma))^{\lambda+1}} \left(\frac{1+\nu}{1-\nu}\right)^3\right)\right) - \log \delta) /\lambda, \delta)$-DP, where  $\sigma \geq \sqrt{|\mathbb{K}|\ln(2+2/\nu)/2n\pi^2}$ and $\kappa$ is defined in Lemma \ref{lem:cont_app}. 
\end{theorem}

Again, given a fixed value of $\delta$, $\varepsilon$ is computed numerically  as in \cite{Abadi,MironovTZ19}.

\subsubsection{Communication overhead} 
\label{sec:com_cost}
The domain of $\mathbf{z}$ in Eq.~\eqref{eq:enc} is the support of $\mathcal{DG}$ which is still unbounded. This means that the
size of the encrypted text can be very large though with exponentially small probability. Indeed, $||\mathbf{z}_k||_{\infty}$ is unbounded and hence 
modulo $m = 2^{\lceil \log_2(\max_k||\mathbf{z}_k||_{\infty} |\mathbb{K}|)\rceil}$ has to be large. To overcome this problem, we choose modulo $m$ to be so large that the probability that $2^{\lceil \log_2(\max_k||\mathbf{z}_k||_{\infty} |\mathbb{K}|)\rceil}$ is larger than $m$ is negligible. For this purpose, we rely on the following concentration inequality of the discrete Gaussian distribution.

\begin{lemma}[\cite{Micciancio017}, Lemma 2.2]
\label{lem:cont}
For any $\nu > 0$, $\xi > \sqrt{\ln(2+2/\nu)/2\pi^2}$, and $t >0$,
$
\Pr_{x \sim \mathcal{DG}(\mu, \xi)}[|x - \mu| \geq t\cdot\xi] \leq 2 e^{- t^2/2} \cdot \frac{1+\nu}{1-\nu}
$.
\end{lemma}
Lemma \ref{lem:cont} implies that if $\xi = \sqrt{n}\sigma/\sqrt{|\mathbb{K}|} > 3.51$ then $\frac{1+\nu}{1-\nu} < \frac{3}{2}$ and 
$
\Pr_{\mathbf{z} \sim \mathcal{DG}(\bm{\mu}, \sqrt{n}\sigma\mathbf{I}/\sqrt{|\mathbb{K}|})}[||\mathbf{z} - \bm{\mu}||_{\infty} \geq t \sqrt{n} \sigma]  
 \leq 3 n e^{-|\mathbb{K}|t^2/2}
$
after applying the union bound.
For example, if $m=2^{\lceil \log_2(12 \sqrt{n}\sigma |\mathbb{K}|)\rceil}$ (i.e., $t=12$) then the probability that $||\mathbf{z}_k -\bm{\mu}_k||_{\infty}$ cannot be bounded by $12 \sqrt{n} \sigma$ per client is less than $2^{-80}$ even if $|\mathbb{K}|=1$ and $n=10^7$. Thus, a client needs to transfer $n\cdot \log_2\left( 2^{\lceil \log_2((12 \sqrt{n} \sigma + \max_k||\bm{\mu}_k||_{\infty}) |\mathbb{K}|)\rceil} \right)$ bits in total to the aggregator. For example, if $|\mathbb{K}|=100$, $\max_k||\bm{\mu}_k||_{\infty}= 1$, $\sigma = 1$ (i.e., $\varepsilon \approx 0.2$), then $\log_2 m =22$. By contrast, if  noise was generated from the continuous domain, then $\log_2 m = 32$ 
which means that DGM reduces the communication overhead by roughly 32\%. 

Notice that if $\varepsilon$ or $\delta$ is smaller (i.e., there is stronger privacy guarantee), then $\sigma$ is larger which implies that $m$ also increases, and hence more bits need to be transferred to the server per parameter. This results in a trade-off between Differential Privacy and bandwidth efficiency.


\begin{algorithm}[t]
\small
		\caption{DP-SignFed: Federated Learning with Client Privacy \label{alg:sign_client_fed_learn}}
	\DontPrintSemicolon
	{\bf Server:}\;
	\Indp Initialize common model $w_0$\;
	\For {$t=1$ \KwTo $\Tcl$}
	{
	    Select $\mathbb{K}$ clients randomly \;
		\For {each client $k$ \textrm{in} $\mathbb{K}$}
		{	
			$\Delta \mathbf{w}_t^k = \mathbf{Client}_k(\mbf{w}_{t-1})$\;
		}
		$\mbf{w}_{t} = \mbf{w}_{t-1} + \gamma \cdot \mathsf{sign}\left(\sum_{k} \Delta \mathbf{w}_t^k\right) $\;
	}
    \Indm {\bf $\mathbf{Client}_{k}(\mbf{w})$:}\;
    \Indp
	$\mbf{w}_{t-1}^k = \mbf{w}$\;
	$\Delta \mbf{w}_t^k = \mathbf{SGD}(D_k, \mbf{w}_t^{k-1}, \Tgd) - \mbf{w}_{t-1}^k$\;
    \KwOut{$\mathsf{Enc}_{K_k}\left(\mathcal{DG}\left(\mathsf{sign}\left(\Delta \mbf{w}_t^k\right), \sqrt{n} \mathbf{I}\sigma /\sqrt{|K|}\right)\right)$}
\end{algorithm}

\subsection{Performance evaluation}
\label{sec:eval}

The performance of DP-SignFed is compared with DP-StdFed in  Table \ref{tab:dp_acc_mnist} and \ref{tab:dp_acc_mnist_fashion}. DP-StdFed is an extension of StdFed to provide client-level differential privacy. Specifically, in DP-StdFed, the randomly selected clients first clip their model update vector to have a bounded $L_2$-norm\footnote{The sensitivity $S=\sqrt{n}$ and $\gamma=0.005$ for DP-FedSign. For DP-FedStd, the server computes the median $L_2$-norm value over $N$ $L_2$-norm values received during an additional initialization round. Hence, $S$ is set to 1.73 and 2.15 for MNIST and Fashion-MNIST, respectively.}, add continuous Gaussian noise to the clipped update vector, and then transfer the \emph{non-quantized} noisy model update to the server (see Alg.~5 in the Appendix of \cite{Our_paper} for more details).
The configurations of these protocols are summarized in Table 17 of \cite{Our_paper}.

Table \ref{tab:dp_acc_mnist} and \ref{tab:dp_acc_mnist_fashion} show the best model accuracy observed over 200 rounds with each algorithm on the MNIST and Fashion-MNIST datasets, respectively. 
DP-StdFed provides the best accuracy; for MNIST, it is 86-93\%, and for Fashion-MNIST, it is 63-78\% depending on the privacy parameter $\varepsilon$. 
The performance degradation of DP-SignFed compared to DP-StdFed is 0.02 on MNIST and 0-0.07 on Fashion-MNIST. 
As expected, weaker privacy requirement (i.e., larger $\varepsilon$) needs smaller noise magnitude and hence better accuracy for all algorithms.

The communication cost of DP-SignFed is 66\% of that of DP-StdFed. Specifically, while DP-StdFed needs 32 bits per parameter, DP-SignFed requires 21-22 bits depending on the value of $\varepsilon$\footnote{It is computed from $\log_2\left( 2^{\lceil \log_2((12 \sqrt{n} \sigma + \max_k||\bm{\mu}_k||_{\infty}) |\mathbb{K}|)\rceil} \right)$ where $\sigma$ is obtained from $\varepsilon$ and $\delta=10^{-5}$ using the moments accountant. This ensures that the magnitude of the noisy update per model parameter is less than the modulus $n$ with probability at most $2^{-80}$ (see Section \ref{sec:com_cost}).}. If $\varepsilon$ is smaller, the variance of the noise is larger, and hence more bits are necessary to encode the noisy signs. Notice that even if we use early stopping and record the best accuracy sooner than 200 rounds, it will neither decrease the communication cost nor increase the privacy guarantee. 
Indeed, one has to execute all the 200 rounds in the first place to identify the best performing round.

%
%
%
%

\begin{table}[!h]
	\caption{Model accuracy and communication cost on MNIST dataset. We give the communication cost per parameter value (bits/parameter) for any value of $\varepsilon$.}
	\label{tab:dp_acc_mnist}
	\centering
	\begin{tabular}{c|c|c|c|c|c|c}
		\hline
		 & \multicolumn{2}{c|}{$\varepsilon=1$}  & \multicolumn{2}{c|}{$\varepsilon=2$}   & \multicolumn{2}{c}{$\varepsilon=4$}\\
		\cline{2-7}
		
		&  Acc & Cost &  Acc & Cost &  Acc & Cost \\
		\cline{2-7}
		DP-StdFed	                & 	0.86 & 32 	    &      0.92 & 32     &  0.93 & 32     \\
		DP-SignFed      	        &   0.87 & 22 	   &  0.90  & 21     &  0.91 & 21   \\
		\hline
	\end{tabular}
\end{table}

\begin{table}[!h]
	\caption{Model accuracy and communication cost with  Fashion-MNIST dataset. We give the communication cost per parameter value (bits/parameter) for any value of $\varepsilon$.}
	\label{tab:dp_acc_mnist_fashion}
	\centering
	\begin{tabular}{c|c|c|c|c|c|c}
		\hline
		 & \multicolumn{2}{c|}{$\varepsilon=1$}  & \multicolumn{2}{c|}{$\varepsilon=2$}   & \multicolumn{2}{c}{$\varepsilon=4$}\\
		\cline{2-7}
		
		&  Acc & Cost &  Acc & Cost &  Acc & Cost \\
		\cline{2-7}
		DP-StdFed		            & 0.63 & 32 	                    &    0.74 & 32      &  0.78 & 32 \\
		DP-SignFed 	                & 0.63 & 22                        &    0.70 & 21      &  0.73 & 21 \\
		\hline
	\end{tabular}
\end{table}


\section{Security Analysis}
\label{sec:security}


This section evaluates the robustness of SignFed, DP-SignFed and DP-StdFed against several state-of-the-art security attacks.

\subsection{Security Model}
\label{sec:sec_attacks}
\noindent \textbf{Adversarial model:} 
In this work, we assume that the adversary controls a certain fraction of the participating entities/clients at each round of the training, which means it can 
access and modify these clients' training data as well as all parameters of their local model. We, however, assume that \emph{the server is honest} (i.e., it
does not manipulate the aggregate or the update vector sent by any client).
The set of all malicious nodes is denoted by $\mathbb{M}$.

We consider two types of adversary. The first one aims at degrading the overall model performance (i.e., increase the average misclassification rate). The 
second one aims at causing targeted misclassification on some particular classes of samples by injecting backdoors into the model during the training phase. 
These adversaries are \emph{active} in the sense that they may not follow the learning protocol  faithfully. 
\smallskip

Next, we detail the attacks considered in our work.

\subsubsection{Overall Model Degradation Attacks}

\paragraph{Random Update Attack}
In this attack, malicious clients, whose numbers might vary as shown later, use random updates.
More specifically, instead of the true model update $\Delta \mbf{w}_t^k$,
each malicious client $k$ generates a random update $\Delta \hat{\mbf{w}}_t^k$ in all time slots $t$ \cite{EPFL_attack}, where $\Delta \hat{\mbf{w}}_t^k$ is drawn from an isotropic Gaussian distribution $\mathcal{G}(0, \sigma_{\Adv} \mathbf{I})$ with mean zero and variance $\sigma_{\Adv}^2$. 
Each malicious party selects the noise independently (i.e., they do not collude).

\paragraph{Gradient Ascent Attack} 
In this attack, malicious clients aim at maximizing the loss by performing gradient \emph{ascent} instead of descent on their own training data. In particular, \emph{every} malicious client $k \in \mathbb{M}$ updates the model parameters locally as 
$
\mbf{w}^k_{\ell} = \mbf{w}^k_{\ell-1} + \eta_{\Adv} \nabla f(\cup_{k\in \mathbb{M}} D_k;\mbf{w})
$,
where $\eta_{\Adv}$ is set in order to suppress the updates of honest clients and to maximize the impact of their own update on the common model. Notice that this attack assumes \emph{colluding malicious clients} (i.e., every malicious client sends exactly the same update computed on the union of their training data).
This attack attempts to maximize the \emph{average} misclassification rate of the common model, and is more effective if the number of malicious parties is large, or the training data of the malicious and benign nodes come from similar distributions. 

We note that Gradient Ascent Attack is equivalent to the Sign Inversion Attack for SignFed, described in \cite{SIGNSGD_vote_robustness}, if $\Tgd = 1$ (i.e., each client computes its update using a single mini-batch in every round).
In Sign Inversion Attack, all malicious clients faithfully compute the sign of their model update, but then send the \emph{inverted} signs to the server for aggregation.  

\subsubsection{Backdoor Attacks (Targeted Attacks)}
The goal of these attacks is to selectively degrade the accuracy of the common model with respect to only a few tasks. As opposed to the overall model degradation attacks, they generate \emph{targeted} misclassification  while preserving the model convergence as well as a high average prediction accuracy except, of course, for the targeted tasks, called backdoor classes.

We distinguish two types of backdoors: In-backdoors and Out-backdoors.

\begin{itemize}

\item[-]{\it In-backdoor Attacks:} In-backdoor attacks \cite{IBM_attack} are created for a class of samples that exists in the training data of some parties. Specifically, for some training samples $D_{\mathit{aux}} \subseteq D_k$, each adversary uses output labels that are different from their true labels. Let $y'$ denote the adversarially chosen label for a training sample $(x, y) \in D_{aux}$, and $D_{aux}'$ denotes the set of all relabelled samples (i.e., $(x, y') \in D_{aux}'$). The new objective is to minimize the loss $f((D_k \setminus D_{aux}) \cup D_{aux}'; \mbf{w})$. 

\item[-]{\it Out-backdoor Attacks:} As opposed to in-backdoors, \emph{out-backdoors} are created from samples that do \emph{not} exist in the training data of any honest clients and are relabelled to have a class that \emph{does} exist in their training data. 
Specifically, let $L$ denote the set of labels that exist in $D = \cup_k D_k$. The adversary creates $D_{aux}''$ such that, for each $(x, y) \in D_{aux}''$, $x \notin D$ and $y \in L$. The new objective is to minimize the loss $f(D_k \cup D_{aux}''; \mbf{w})$. 

\end{itemize}

To illustrate the difference between in- and out-backdoors, consider a model which recognizes dogs and rabbits in the input photos. If the adversary relabels all photos of dogs as 'rabbit' in its training data, then it is an in-backdoor attack. However, if the adversary adds new photos of frogs to its training data and relabels them as 'dog', then this is an out-backdoor attack. 

Out-backdoors are more difficult to detect than in-backdoors as they can come from a much larger set of samples, which are potentially unknown to the protocol participants.  
Hence, out-backdoors are especially severe in security-related applications such as in access control. 

As per \cite{IBM_attack}, the adversary also uses explicit boosting to outbalance the combined effect of benign model updates.
For both in- and out-backdoors, the adversary boosts $\Delta \mbf{w}_{\Adv}^{t}$ at time $t$ by sending $\eta_{\Adv} \Delta \mbf{w}_{\Adv}^{t}$  ($\mu > 1$) in order to suppress the model updates of benign parties. Importantly, $\eta_{\Adv}$ should be large enough in order to achieve misclassification of the backdoor class but also small enough to ensure the convergence of the common model and hence hide the attack.


\subsection{SignFed Security Analysis}
\label{Sec:sec}

In this section, we evaluate the robustness of SignFed against the security attacks presented previously.  


For the Overall Model Degradation attacks, different percentage of malicious nodes are considered, the MNIST and IMDB datasets were used, and the same experimental setting as defined in Table \ref{tab:fl_sign_param} is used. The boosting parameter $\eta_{\Adv}$ of the Gradient Ascent Attack is set to 10 with MNIST dataset and 20 with the IMDB dataset (we use the boosting only with StdFed). We also do not need to use boosting for the Random Update Attack with StdFed as $\sigma_{\mathsf{Adv}} = 200$ generates large noise which prevents the model convergence.

For the Backdoor attacks, the MNIST and CIFAR datasets were used and the experimental setting is shown in Table \ref{tab:fl_sign_param_backdoor}. As backdoor attacks, which aim at modifying the prediction of one
particular label while maintaining the global accuracy, are more difficult to perform on binary classifiers, we switched to the CIFAR dataset with a multiclass classifier. Furthermore, similarly to \cite{IBM_attack}, we reduce the total number of clients $N$ from 1000 to 10, we use different percentages of malicious nodes: 10\%, 20\% and 40\%, and all clients report their updates to the server at each round (i.e.  $C=1.0$). The malicious nodes collude by sharing their data for the training and by sending the same update to the server.



\subsubsection{Overall Model Degradation Attacks}

\paragraph{Random  update}

Table \ref{tab:random_update_mnist} and \ref{tab:random_update_imdb} depict the best accuracy of the global model over 100 rounds according to the fraction of malicious nodes in set $\mathbb{K}$. The results show that SignFed is robust against the random update attack even if $20$\% of all nodes are malicious, while StdFed fails to converge even if $1$\% of all nodes are malicious. Indeed, with 20\% of malicious nodes, SignFed reaches an accuracy of 98\% and 86\% for the MNIST and IMDB datasets, respectively. On the contrary, StdFed fails to converge even with one malicious node at each round.  In fact, as we show in Appendix \ref{app:conv}, SignFed's convergence rate is $O\left(\frac{1}{(1-\alpha)\sqrt{CN}\Tcl}\right)$, where $\alpha$ denotes the fraction of malicious clients. 
This is in contrast to the sign inversion attack detailed in \cite{SIGNSGD_vote_robustness}, which has a convergence rate of $O\left(\frac{1}{(1-2\alpha)\sqrt{CN}\Tcl}\right)$, that is, convergence is only possible if less than half of the nodes are malicious. 

\begin{table}[!h]
	\caption{Random update attack on SignFed and StdFed with the MNIST dataset depending on the fraction of malicious nodes. $\sigma_{\mathsf{Adv}} = 200$. The table represents the best accuracy over 100 rounds. "-" means that the algorithm does not converge.}
	\label{tab:random_update_mnist}
	\centering
	\begin{tabular}{ccccc}
		\hline
		         			        & 10\%                              & 20\%                      & 40\%                  & 60\% 		             \\
		\hline
		StdFed	                & -                                 & -                         & -		                & -                      \\
		SignFed	                & 0.98		                        & 0.98                      & 0.94                  & -                  \\
		\hline
	\end{tabular}
\end{table}

\begin{table}[!h]
	\caption{Random update attack on SignFed and StdFed with the IMDB dataset depending on the fraction of malicious nodes. $\sigma_{\mathsf{Adv}} = 200$. The table represents the best accuracy over 100 rounds.  "-" means that the algorithm does not converge.}
	\label{tab:random_update_imdb}
	\centering
	\begin{tabular}{ccccc}
		\hline
		         			        & 10\%                              & 20\%                      & 40\%                  & 60\% 		             \\
		\hline
		StdFed	                & -                                 & -                         & -		                & -                      \\
		SignFed	                & 0.86		                        & 0.86                      & 0.54                  & -                  \\
		\hline
	\end{tabular}
\end{table}

\paragraph{Gradient Ascent Attack} \label{sec:overall_attack
}
Table \ref{tab:overall_mnist} and \ref{tab:overall_imdb} show the best accuracy of the global model over 100 rounds when the adversary aims to degrade the average model performance by performing gradient ascent on its own training data. With the MNIST dataset (in Table \ref{tab:overall_mnist}), SignFed reaches an accuracy of 98\% and 79\% for 20\% and 40\% of malicious nodes, respectively, while StdFed does not converge even if only 10\% of the nodes are malicious.
For IMDB dataset, SignFed reaches an accuracy of 86\% and 72\% for 10\% and 20\% of malicious nodes, respectively, while StdFed fails to converge with only 10\% of malicious nodes. Indeed, StdFed does not converge even if we have only one malicious node. The reason for this difference is that malicious nodes can scale up their update with $\eta_{\Adv}$ and hence boost its effect on the global model. However, such adversarial boosting does not work with SignFed as the trusted server accepts only the values $-1$ and $+1$ in the update vectors. Therefore, a single malicious client does not have larger impact on the global model than any other honest client. To boost its impact, the adversary can only increase the number of the malicious clients, as shown by the experimental results. 
Since Gradient Ascent is equivalent to Sign Inversion Attack if $\Tgd=1$, the convergence rate of Gradient Ascent in this restricted scenario is $O\left(\frac{1}{(1-2\alpha)\sqrt{CN}\Tcl}\right)$ as shown in \cite{SIGNSGD_vote_robustness}.


\begin{table}[h]
	\caption{Gradient ascent attack on SignFed and StdFed with the MNIST dataset. $\eta_{\Adv} = 10$. The table represents the best accuracy over 100 rounds. "-" means that the algorithm does not converge.}
	\label{tab:overall_mnist}
	\centering
	\begin{tabular}{ccccc}
		\hline
		         			        & 10\%                              & 20\%                      & 40\%                  & 60\% 		             \\
		\hline
		StdFed	                & -                              & -                         & -		                & -                      \\
		SignFed	                & 0.98		                        & 0.98                      & 0.79                  & -                      \\
		\hline
	\end{tabular}
\end{table}

\begin{table}[h]
	\caption{Gradient ascent attack on SignFed and StdFed with the IMDB dataset.  $\eta_{\Adv} = 20$. The table represents the best accuracy over 100 rounds. "-" means that the algorithm does not converge.}
	\label{tab:overall_imdb}
	\centering
	\begin{tabular}{ccccc}
		\hline
		         			        & 10\%                              & 20\%                      & 40\%                  & 60\% 		             \\
		\hline
		StdFed	                & -                              & -                      & -		                & -                      \\
		SignFed	                & 0.86		                        & 0.72                      & 0.52                  & -                      \\
		\hline
	\end{tabular}
\end{table}

\subsubsection{Backdoor attacks}


\paragraph{In-backdoor attack}

Figure \ref{fig:in_backdoor_mnist}, \ref{fig:in_backdoor_cifar} and Table~\ref{tab:in-Backdoors-MNIST}, \ref{tab:in-Backdoors-CIFAR} show the effect of  in-backdoor attacks on the MNIST and CIFAR datasets, respectively. In all experiments, there are ten clients, out of which different fraction of malicious nodes are considered.
Figure \ref{fig:in_backdoor_mnist} depicts the accuracy of the global model for MNIST, when the adversary relabels every image of digit '5' to '7' in its local dataset. The red plots show the accuracy of the global models, while the green ones display the model accuracy only for the images with label '5' (i.e., accuracy on the backdoor class). The results show that SignFed is robust as both global model accuracy and model accuracy on the specific in-backdoor class (digit 5) reach 99\% by the end of the training. 

By contrast, with StdFed, while the accuracy of the global model converges slowly to 99\%, the accuracy of the attacked model oscillates. Similar behaviour can be observed in Figure \ref{fig:in_backdoor_cifar} which plots the accuracy on CIFAR dataset, where images of airplanes are re-labelled to 'ship' in the adversary's training data. In these experiments, StdFed fails to converge on the backdoor class, and its accuracy  on CIFAR never exceeds 55\%.

The oscillation of accuracy with StdFed can be explained by the nature of gradient descent and in particular backpropagation: when the malicious client injects the backdoor, it scales its update with $\eta_{\Adv}$. In the following round, honest clients scale up their gradients on the backdoor samples (i.e., images of digit 5 in MNIST and images of airplanes in CIFAR) in order to ``fix'' the classification error on the backdoor class. In the next round, when the model is ``fixed'' (i.e., digit '5' is correctly predicted as '5' again), the adversary's gradients are increased again in order to re-inject the backdoor. This process repeats till the end of the training.
By contrast, and similarly to the overall model degradation attacks, a malicious client cannot scale up its update in SignFed as the update vectors must take value from $\{-1,1\}^n$. 

Table~\ref{tab:in-Backdoors-MNIST} shows the accuracy of the model on digit class 5 (in-backdoor class) when we consider different percentage of malicious nodes (values are chosen based on the best model accuracy over 40 rounds). The global accuracy of the model over all the classes is 99\% and the accuracy on class '5' is 99\% independently of  the number of malicious nodes and regardless whether StdFed or SignFed is used. 

As in the previous table, Table~\ref{tab:in-Backdoors-CIFAR} shows the accuracy of the model on airplane class (in-backdoor class) when we consider different number of malicious nodes (values are chosen based on the best global accuracy over 100 rounds). 
SignFed with 20\% of malicious nodes reaches a global accuracy of  84\%, and an  accuracy of 76\% on the in-backdoor class. However, StdFed with the same amount of malicious nodes reaches 80\% of global accuracy and 0\% for the airplane class. The results confirm the larger robustness of SignFed over StdFed.

\begin{figure}[h]
  \centering
  \includegraphics[scale=0.5]{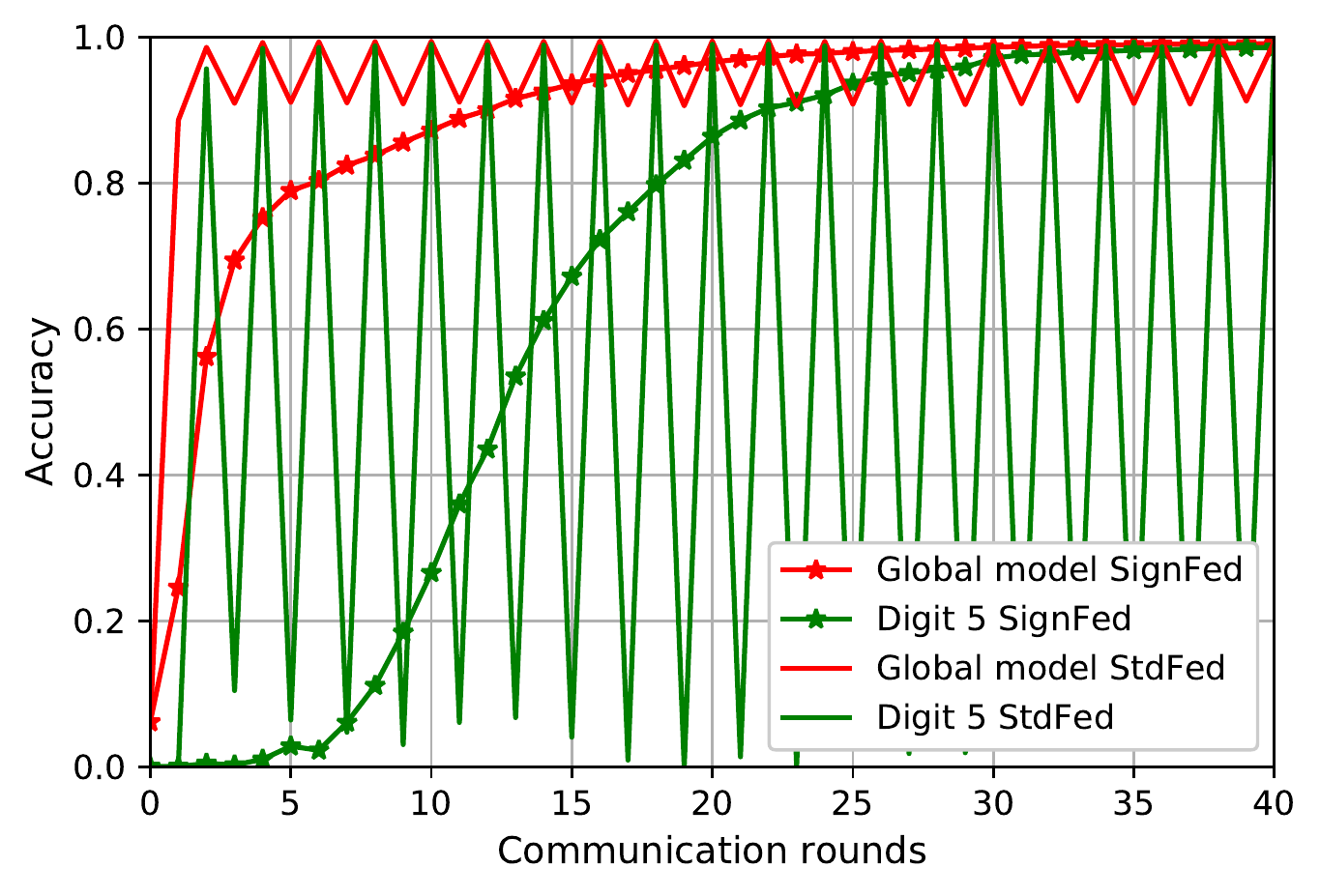}
  \caption{In-backdoor attack on SignFed and StdFed with the MNIST dataset, $\eta_{\Adv} =7$. The figure displays the global accuracy convergence and the accuracy of the label "5" which is under attack. 10\% of the nodes are malicious.}
  \label{fig:in_backdoor_mnist}
\end{figure}

\begin{figure}[h]
  \centering
  \includegraphics[scale=0.5]{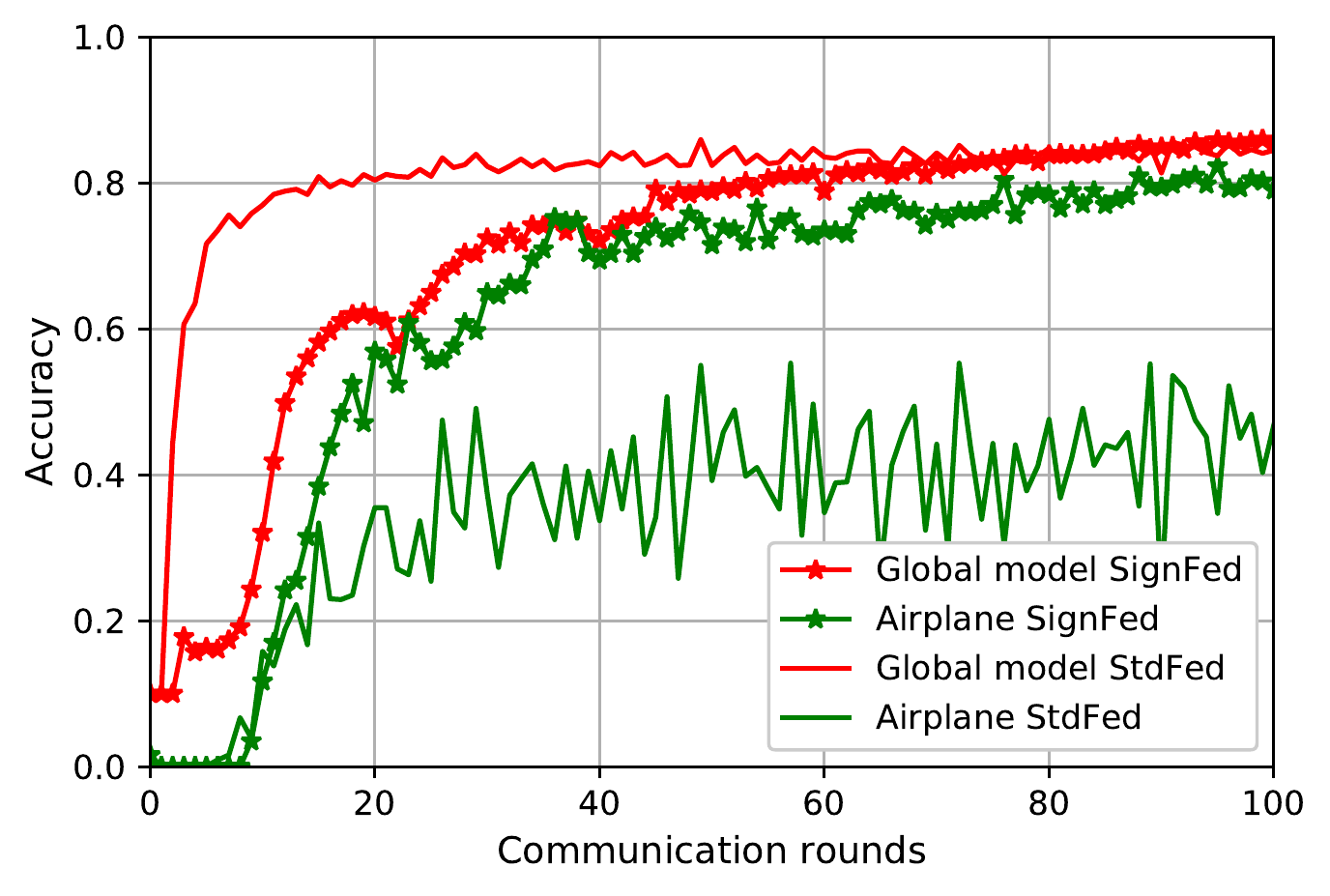}
  \caption{In-backdoor attack on SignFed and StdFed with CIFAR dataset, $\eta_{\Adv} =7$. The figure displays the global accuracy convergence and the accuracy of the label "airplane" which is under attack. 10\% of the nodes are malicious. }
  \label{fig:in_backdoor_cifar}
\end{figure}

\begin{table}[!h]
	\caption{In-backdoor attack on SignFed and StdFed with the MNIST dataset, $\eta_{\Adv}$ is set to 7, 3, 1 for 10\%, 20\%, 40\% respectively. The table depicts the global accuracy convergence and the accuracy of the label "5" which is under attack.} 
	\label{tab:in-Backdoors-MNIST}
	\centering
	\begin{tabular}{ccccc}
		\hline
		         			             &       & $10\%$  		                    & $20\%$         & $40\%$ \\
		\hline
\multirow{2}{*}{\emph{SignFed}} &		Model accuracy	                        & 0.99		    &     0.99          &    0.99        \\
	                            &	Accuracy on digit class '5'	                        &       0.99     &   0.99        &    0.99         \\

		\hline
		\hline
\multirow{2}{*}{\emph{StdFed}} &		Model accuracy	                        & 	0.99	    &     0.99          &    0.99        \\
	                            &	Accuracy on digit class '5'	                        &      0.99      &   0.99        &    0.99         \\

		\hline		

	\end{tabular}
\end{table}

\begin{table}[!h]
	\caption{In-backdoor attack on SignFed and StdFed with the CIFAR dataset, $\eta_{\Adv}$ is set to 7, 4, 2 for 10\%, 20\%, 40\% respectively. The table depicts the global accuracy convergence and the accuracy of the label "airplane" which is under attack.} 
	\label{tab:in-Backdoors-CIFAR}
	\centering
	\begin{tabular}{ccccc}
		\hline
		         			             &       & $10\%$  		                    & $20\%$         & $40\%$ \\
		\hline
\multirow{2}{*}{\emph{SignFed}} &		Model accuracy	                        & 	0.86	    &     0.84          &    0.82        \\
	                            &	Accuracy on airplane class	                        &     0.80       &   0.76        &    0.57         \\

		\hline
		\hline
\multirow{2}{*}{\emph{StdFed}} &		Model accuracy	                        & 	0.86	    &     0.80          &    0.81        \\
	                            &	Accuracy on airplane class	                        &     0.55       &   0        &    0         \\

		\hline		

	\end{tabular}
\end{table}


\paragraph{Out-backdoor attack}

The main goal of the out-backdoor attack is to introduce fake information during the training by relabeling a sample, whose true label is not a valid output of the global model. We experimented this attack on MNIST by first excluding all samples with digit '0' in all clients' training datasets. We then choose different fraction of malicious clients and relabeled the samples with '0' to '1'.  Similarly, the attack is also implemented using the CIFAR dataset by removing all airplanes from the clients' training data and relabelling all images of an airplane as 'ship' in the malicious clients' datasets \footnote{we also removed all birds and trucks from the training data, in order to limit the bias between classes.}. Note that since only malicious clients have samples from the backdoor class, the detection of this attack is quite challenging.  

Tables \ref{tab:outbackdoor_mnist} and \ref{tab:outbackdoor_cifar} display the global model accuracy as well as the model's prediction rate to misclassify the out-backdoor class to the targeted class (attack accuracy) for MNIST and CIFAR, respectively (values are chosen based on the best model accuracy over 100 rounds with MNIST and 300 rounds with CIFAR). We consider different fraction of malicious nodes. The results show that the model accuracy is similar for both datasets and schemes, but SignFed is much more robust against the attacks than StdFed. In fact, with 10\% of malicious nodes, the attack accuracy on the MNIST dataset is very low for SignFed (19\%) whereas it is quite large for StdFed (92\%). We obtained similar pattern with the CIFAR dataset although the accuracy difference is less significant (66\% versus 72\%). This can be explained by the inherent bias present in CIFAR. For example, planes are often misclassified as 'bird' or 'ship' even without the attack because of the similar background of these images (i.e., sky is very similar to sea in many images).  Indeed, the probability of predicting an airplane as a ship without the attack is 58\%, and it only increases to 66\% and 72\% with SignFed and StdFed, respectively. 

As for in-backdoor attacks, SignFed mitigates out-backdoor attacks because the adversary cannot scale up its update in order to increase its impact on the global model.

\begin{table}[h]
	\caption{Out-backdoors attack on SignFed and StdFed with the MNIST dataset. $\eta_{\mathsf{Adv}}$ is set to 1 (no boosting). The table displays the global model accuracy as well as the model’s prediction rate to misclassify the out-backdoor class "0" to the targeted class "1" (attack accuracy).}
	\label{tab:outbackdoor_mnist}
	\centering
	\begin{tabular}{ccccc}
		\hline
		         			             &       & $10\%$  		                    & $20\%$         & $40\%$ \\
		\hline
\multirow{2}{*}{\emph{SignFed}} &		Model accuracy	                        & 	0.99	    &     0.99          &    0.99        \\
	                            &	Attack accuracy	                        &        0.19    &   0.87        &    0.99         \\

		\hline
		\hline
\multirow{2}{*}{\emph{StdFed}} &		Model accuracy	                        & 	0.99	    &     0.99          &    0.99        \\
	                            &	Attack accuracy	                        &       0.92     &   0.99        &    0.99        \\

		\hline		

	\end{tabular}
\end{table}

\begin{table}[h]
	\caption{Out-backdoors attack on SignFed and StdFed with the CIFAR dataset. $\eta_{\mathsf{Adv}}$ is set to 1 (no boosting). The table displays the global model accuracy as well as the model’s prediction rate to misclassify the out-backdoor class "airplane" to the targeted class "ship" (attack accuracy).}
	\label{tab:outbackdoor_cifar}
	\centering
	\begin{tabular}{ccccc}
		\hline
		         			             &       & $10\%$  		                    & $20\%$         & $40\%$ \\
		\hline
\multirow{2}{*}{\emph{SignFed}} &		Model accuracy	                        & 	0.91	    &     0.91          &    0.92        \\
	                            &	Attack accuracy	                        &        0.66    &   0.74        &    0.93         \\

		\hline
		\hline
\multirow{2}{*}{\emph{StdFed}} &		Model accuracy	                        & 	0.92	    &     0.92          &    0.90        \\
	                            &	Attack accuracy	                        &       0.72     &   0.86        &    0.95        \\

		\hline		

	\end{tabular}
\end{table}




\subsection{DP-SignFed Security Analysis}
\label{sec:sec_dp}

Table \ref{tab:dp-in-Backdoors-MNIST} and \ref{tab:dp-in-Backdoors-Fashion-MNIST} depict the accuracy of the in-backdoor class and the misclassification rate of the out-backdoor class (best values are chosen based on the global model accuracy over 200 rounds) when backdoor attacks are launched against our DP schemes.

Malicious clients, whose fraction changes between 0.1 and 0.4, omit to add noise to their own updates at each round.
We use the configuration described in Table \ref{tab:fl_sign_param_privacy} except for $\gamma$ which is decreased to 0.001, and $\delta$ is fixed to $10^{-5}$. In Fashion-MNIST dataset, at all malicious nodes, all images of 'Sandal' are relabelled to 'Sneaker' for In-backdoor,  and all images of 'T-shirt/top' are relabelled to 'Trouser' for Out-backdoor attacks (only the malicious nodes have photos of 'T-shirt/top'). In addition, with DP-SignFed, each malicious node calculates their updates, extracts the signs ($\mathsf{sign} : \mathbb{R}^n \rightarrow \{ -1,1\}^n$) and then uses a boosting parameter $\eta_{\Adv}=5000$ to boost their updates before sending them back to the server for aggregation. Indeed, as all honest clients send the noisy update in DP-SignFed, the noise together with encryption can conceal the manipulation of the malicious update vectors.

The results show that DP-SignFed are less robust against backdoor attacks than SignFed. On the MNIST dataset, model accuracy on the in-backdoor class is 0\% for DP-SignFed regardless of the number of malicious nodes, and larger than 97\% and 95\% for SignFed, with 10\% and 20\% of malicious nodes, respectively. The same tendency holds for Fashion-MNIST. Out-backdoor attacks are especially effective on MNIST (see Table~\ref{tab:dp-out-backdoor_mnist} and Table~\ref{tab:dp-out-backdoor_fashion_mnist}); here, the misclassification rate is more than 98\% for DP-SignFed and 0-99\% for SignFed. When we consider only 2\% of malicious nodes with MNIST, the misclassification rate is 0\% for SignFed and 76\% for StdFed (without boosting) with a global model accuracy of 98\% for both schemes. Indeed, StdFed is vulnerable to the outbackdoor attack even if we have only a small number malicious node and without using any boosting. For MNIST and Fashion-MNIST, SignFed is clearly superior to DP-SignFed regarding all attacks.

Finally, random update attack and gradient ascent attack are mounted against DP-SignFed. The same parameters are used as in the previous experiments. Malicious clients still omit to add any noise to their own model updates. Instead, they boost their signs updates with DP-SignFed ($\eta_{\mathsf{Adv}} = 5000$). The model fails to converge even if only 1\% of all selected nodes are  malicious at each round.

\begin{table}[!h]
	\caption{In-backdoor attack on DP-SignFed and SignFed with MNIST dataset.  
     The table depicts the global accuracy convergence and the accuracy of the label "5" which is under attack.} 
	\label{tab:dp-in-Backdoors-MNIST}
	\centering
	\begin{tabular}{ccccc}
		\hline
		         			             &       & $10\%$  		                    & $20\%$         & $40\%$ \\
		\hline
\multirow{2}{*}{$\varepsilon=1$} &		Model accuracy                       & 	0.89	    &       0.90        &       0.90     \\
	                            &	Accuracy on digit class '5'	             &    0       &       0    &         0   \\

		\hline
		\hline
\multirow{2}{*}{\emph{$\varepsilon=2$}} &		Model accuracy	                        & 	0.89	    &     0.90 &      0.90      \\
	                            &	Accuracy on digit class '5'	                        &      0      &        0   &       0      \\

		\hline	
		
		\hline
		\hline
\multirow{2}{*}{\emph{$\varepsilon=4$}} &		Model accuracy	                        & 	0.89	    &     0.90          &   0.90         \\
	                            &	Accuracy on digit class '5'	                        &   0     &       0    &     0        \\

		\hline		
		
		\hline
		\hline
\multirow{2}{*}{\emph{SignFed}} &		Model accuracy	                        & 	0.98	    &   0.98            &   0.90         \\
	                            &	Accuracy on digit class '5'	                        &   0.97      &    0.95       &       0      \\

		\hline		

	\end{tabular}
\end{table}

\begin{table}[!h]
	\caption{In-backdoor attack on SignFed and DP-SignFed with Fashion-MNIST dataset. 
	The table depicts the global accuracy convergence and the accuracy of the label "Sandal" which is under attack. } 
	\label{tab:dp-in-Backdoors-Fashion-MNIST}
	\centering
	\begin{tabular}{ccccc}
		\hline
		         			             &       & $10\%$  		                    & $20\%$         & $40\%$ \\
		\hline
\multirow{2}{*}{$\varepsilon=1$} &		Model accuracy                       & 	0.77	    &       0.79   &        0.80    \\
	                            &	Accuracy on Sandal class	            &     0      &      0     &        0    \\

		\hline
		\hline
\multirow{2}{*}{\emph{$\varepsilon=2$}} &		Model accuracy	                        & 	0.77	    &     0.79 &      0.80      \\
	                            &	Accuracy on Sandal class	                        &      0      &      0     &     0        \\

		\hline	
		
		\hline
		\hline
\multirow{2}{*}{\emph{$\varepsilon=4$}} &		Model accuracy	                        & 	0.77	    &    0.79  &     0.80       \\
	                            &	Accuracy on Sandal class	                        &      0      &     0      &      0       \\

		\hline		
		
		\hline
		\hline
\multirow{2}{*}{\emph{SignFed}} &		Model accuracy	                        & 	0.83	    &   0.84       &      0.79      \\
	                            &	Accuracy on Sandal class	                        &   0.90       &    0.84       &        0     \\

		\hline		

	\end{tabular}
\end{table}

\begin{table}[h]
	\caption{Out-backdoors attack on DP-SignFed and SignFed with MNIST dataset. The table displays the global model accuracy as well as the model’s prediction rate to misclassify the out-backdoor class "0" to the targeted class "1" (attack accuracy).}
	\label{tab:dp-out-backdoor_mnist}
	\centering
	\begin{tabular}{cccccc}
		\hline
		         			             &               &     $2\%$                     & $10\%$  		                    & $20\%$         & $40\%$ \\
		\hline
\multirow{2}{*}{$\varepsilon=1$} &		Model accuracy	 &        0.98               & 	0.98	    &    0.99    &    0.99        \\
	                            &	Attack accuracy	      &       0.98            &      0.99      &       0.99    &       1      \\

		\hline
		\hline
\multirow{2}{*}{$\varepsilon=2$} &		Model accuracy	&        0.98                & 	0.98	    &   0.98            &  0.99          \\
	                            &	Attack accuracy	     &       0.99            &     0.98       &      0.99    &       0.99     \\
	                            
		\hline
		\hline
\multirow{2}{*}{$\varepsilon=4$} &		Model accuracy	  &      0.98                & 	0.98	    &      0.99         &     0.99       \\
	                            &	Attack accuracy	      &      0.99            &      0.99      &       0.99    &     0.99       \\
	                            
		\hline
		\hline
\multirow{2}{*}{SignFed} &		Model accuracy &   0.98	               & 		0.98    &        0.98    &   0.99         \\
	                            &	Attack accuracy	       &     0            &     0.97       &    0.99       &     0.99       \\

		\hline

	\end{tabular}
\end{table}

\begin{table}[h]
	\caption{Out-backdoors attack on DP-SignFed and FL-SIG with Fashion-MNIST dataset. The table displays the global model accuracy as well as the model’s prediction rate to misclassify the out-backdoor class "T-shirt/Top" to the targeted class "Trouser" (attack accuracy).}
	\label{tab:dp-out-backdoor_fashion_mnist}
	\centering
	\begin{tabular}{ccccc}
		\hline
		         			             &       & $10\%$  		                    & $20\%$         & $40\%$ \\
		\hline
\multirow{2}{*}{$\varepsilon=1$} &		Model accuracy	                   & 	0.87	    &    0.88           &   0.90         \\
	                            &	Attack accuracy	                        &     0.78       &    0.81       &     0.87        \\

		\hline
		\hline
\multirow{2}{*}{$\varepsilon=2$} &		Model accuracy	                     & 	0.88	    &     0.88          &   0.90         \\
	                            &	Attack accuracy	                        &   0.78         &    0.82       &       0.85     \\
	                            
		\hline
		\hline
\multirow{2}{*}{$\varepsilon=4$} &		Model accuracy	                   & 	0.87	    &      0.89         &  0.90          \\
	                            &	Attack accuracy	                        &   0.78         &     0.81      &     0.86       \\
	                            
		\hline
		\hline
\multirow{2}{*}{SignFed} &		Model accuracy	               & 	0.87	    &      0.88         &    0.90        \\
	                            &	Attack accuracy	                        &     0       &      0.12     &      0.83      \\

		\hline

	\end{tabular}
\end{table}



\section{Related work}
\label{sec:related_work}

\noindent \textbf{Security of Federated Learning:}
Federated learning being a relatively new concept, its security has not been studied so far. However, most ML security attacks also apply to it. In this section, we mostly focus on integrity attacks \cite{PapernotMSW18}. These attacks include {\em pollution} and {\em backdoor attacks.}

In pollution attacks, the adversary manipulates its training data or the corresponding labels to poison the global model. In backdoor attacks, the adversary manipulates its training data or the incoming labels in order to insert backdoors into the global model. Indeed, the goal of the adversary is to cause the misclassification of  specific labels in the global model while maintaining a good accuracy on the non-targeted labels. In the context of Federated Learning, the most efficient strategy consists of directly manipulating the model updates. 

The first backdoor attack designed for a federated learning environment was proposed in \cite{Backdoor_Fed}. Here, the adversary scales up its update in order to surpass the contributions of other honest participants after aggregation. The goal of the attack is to alter the common model so that it exhibits some adversarial behaviour (e.g., targeted misclassification). However, these attacks are effective only in later rounds, when the global model has converged. Indeed, the attack exploits the fact that when the global model has converged, the updates of other honest clients will be smaller and then are more easier to surpass. In contrast, the adversary in \cite{IBM_attack} boosts its update enough to surpass the contributions of the honest clients from the very first rounds even when the global model has not converged.

The resilience of distributed implementations of Stochastic Gradient Descent (SGD) against Byzantine failures is studied in \cite{EPFL_attack}. Each Byzantine worker (among a set of workers) sends a random vector drawn from a Gaussian distribution. The results show that only a single Byzantine worker can prevent the traditional federated schemes such as StdFed from converging. In the same paper, KRUM a Byzantine-resilient algorithm is proposed  as an aggregation rule to select one honest update per round in an adversarial environment.  

In \cite{SIGNSGD_vote_robustness}, the authors study the robustness and the tolerance of signSGD/SIGNUM \cite{SIGNSGD} with majority vote  against network faults and adversarial clients, where SIGNUM is the momentum equivalent of signSGD (i.e., each client maintains a momentum and transmits the sign momentum to the server at each iteration). 
In \cite{SIGNSGD_vote_robustness}, the authors show that signSGD is robust against sign inversion attack, when each malicious client inverts the sign of the computed gradient. The authors argue that this is the best possible attack in a \emph{non-adaptive} setting (i.e., when the adversary performs the attack independently of the gradients it computed). In this paper, we experimentally show that SignFed is also robust against other adaptive attacks like various backdoor attacks \cite{IBM_attack}. \smallskip 

\noindent \textbf{Privacy of Federated Learning:}
There exist a few inference attacks specifically designed against federated learning schemes. In \cite{Property_inference}, the adversary's goal is to infer whether records with a specific property are included in the training dataset of the other participants (called batch property inference). The authors demonstrate the attack by inferring whether black people are included in any of the training datasets, where the common model is trained for gender classification (i.e., the inferred property is independent of the learning objective). The adversary is supposed to have access to the aggregated model update of honest participants.
In \cite{NasrSH19}, the proposed attack infers if a specific person is included in the training dataset of the participants (aka, Membership inference). The adversary extracts the following features from every snapshot of the common model, which is a neural network: output value, hidden layers, loss values, and the gradient of the loss with respect to the parameters of each layer. These features are used to train a membership inference model, which is a convolutional neural network.

The concept of Client-based Differential Privacy has been introduced in \cite{Client-DP-McMahan} and \cite{Client-DP-ETH-Zurich}, where the goal is to hide any information that is specific to a single client's training data. These algorithms bound and noise the contribution of a single client's instead of a single record in the client's dataset. The noise is added by the server, hence, unlike our solution, these works assume that the server is trusted. Also, the noise is drawn from continuous distributions. \smallskip

\noindent \textbf{Bandwidth Optimization in Federated Learning:}
Different quantization methods have been proposed to save the bandwidth and reduce the communication costs in federated learning. They can be divided into two main groups: unbiased and biased methods. The unbiased approximation techniques use probabilistic quantization schemes to compress the stochastic gradient and attempt to approximate the true gradient value as much as possible \cite{QSGD}\cite{TERNGRAD}\cite{ATOMO}\cite{Quant_Fed}. 
However, biased approximations of the stochastic gradient can still guarantee convergence both in theory and practice \cite{SIGNSGD, LinHM0D18, SeideFDLY14}. In signSGD \cite{SIGNSGD}, all the clients calculate the stochastic gradient based on a single mini-batch and then send the sign vector of this gradient to the server. The server calculates the aggregated sign vector by taking the median (majority vote) and sends the signs of the aggregated signs back to each client.

The main differences between our scheme (SignFed) and signSGD are as follows:

\begin{itemize}
    \item SignFed aims to train a common model that is distributed to a random subset of all clients in every round. However, in signSGD, each client builds its own model locally and the server sends the same aggregated model update to \emph{every} client. Selecting only a random subset of clients in each round has at least three benefits. First, SignFed becomes more robust against temporary node failures. Second, SignFed reduces the communication costs upstream to the server. Finally, sampling boosts privacy due to the uncertainty that a specific user's or client's data is used for training or not.
    \item In SignFed, each client can perform multiple SGD iterations locally using multiple mini-batches before computing the model update. On the contrary, signSGD always performs one local SGD iteration with a single mini-batch at every client. This is needed to guarantee convergence since all nodes maintain different local models unlike in SignFed.
    \item As there is no single common model built in signSGD, the server only transfers the sign of the aggregated signs to the clients in every round. Therefore, only a single bit is transferred per parameter downstream to the clients. In SignFed, the whole model is transferred but only to a random subset of clients. 
 
\end{itemize}



\section{Summary and Discussion}
\label{sec:discussion}

  

We can make the following main observations.
\begin{enumerate}
    \item SignFed is almost as accurate as StdFed but incurs less communication overhead and has better resiliency against both security and privacy attacks (see Section \ref{sec:fl_sign_per} and \ref{Sec:sec}).
    \item Although SignFed is more robust against state-of-the-art privacy attacks than StdFed, DP-SignFed provides provable privacy guarantees unlike SignFed. However, it also produces models with slightly worse accuracy than SignFed. More importantly, it is less robust against security attacks (see Section \ref{sec:eval} and \ref{sec:sec_dp}).
    
    \item DP-SignFed has 30-40\% less communication cost than StdFed but it is roughly 20 times more than that of SignFed. In DP-SignFed, there is a trade-off between privacy and bandwidth. Stronger privacy requires to increase the variance of the added discrete Gaussian noise which in turn  implies larger communication costs (see Section \ref{sec:fl_client_dp}).
    \item The convergence rates of SignFed and DP-SignFed are $O\left( \frac{1}{\sqrt{\Tcl C N}} \right)$ and $O\left(\frac{1}{\sqrt{\Tcl C N}} + \frac{n^{3/2}\sigma}{\sqrt{\Tcl} C N}\right)$, respectively, supposing that $\gamma = O(1/\sqrt{\Tcl})$, $\Tgd = 1$, $|\mathbb{B}| = \Tcl$ (see Appendix \ref{app:conv} for the proofs). Therefore, the ``cost of privacy'' in convergence rate is $O\left(\frac{n^{3/2}\sigma}{\sqrt{\Tcl} C N}\right)$ which is due to the added noise.
    
\end{enumerate}



Seemingly, there is a possible trade-off between differential privacy and robustness against security attacks. One possible explanation is that differential privacy requires to randomize every value of the update vector so much that their aggregates become easier to manipulate.  As malicious clients omit to add any noise to their own model updates, the attacked DP protocols essentially turn into Random Update Attacks, where honest clients send almost uniformly random signs and malicious clients transfer non-noisy, boosted updates to the aggregator. As SignFed converges with Random Update Attack even with limited number of honest nodes, the malicious nodes in DP-SignFed can also degrade model performance or inject backdoors for the very same reason.       
The smaller $\varepsilon$ is the more uniform every coordinate's distribution will be, and the larger impact a malicious client has on the aggregate.

\bibliographystyle{ACM-Reference-Format}
\bibliography{references}

\appendix

\section{Appendix}

\subsection{Proof of Lemma \ref{lem:cont_app}}
\label{app:proof_lem_cont_app}

\begin{proof}
We first show that $Z/\sqrt{2\pi}\xi \leq  1 + \frac{2e^{-2\pi^2\xi^2}}{1-e^{-6\pi^2\xi^2}}$, where $Z = \sum_{x \in \mathbb{Z}} \exp(- (x-\mu)^2/2\xi^2)$,
which implies the upper bound.
From \cite{SZABLOWSKI2001289}, $Z = \sqrt{2\pi}\xi \vartheta_3(\pi \mu, \exp(-2\pi^2\xi^2))$, where $\vartheta_3(u, r) = 1 + 2\sum_{i \geq 1} r^{i^2} \cos(2iu)$ is a Jacobi Theta function. Then,
\begin{align*}
1 + 2\sum_{i \geq 1} r^{i^2} \cos(2iu) &\leq 1 + 2 r\sum_{i \geq 0} r^{3i} \\
& \leq 1 + \frac{2r}{1-r^3} 
\end{align*}
if $|r|<1$. The lower bound can be derived similarly using the fact that $\cos(2iu) \geq -1$.
\end{proof}

\subsection{Proof of Theorem \ref{thm:dg_privacy}}
\label{app:proof_thm_dg_privacy}

\begin{proof}
Without loss of generality, suppose that $\mathcal{DG}_{\xi}: \mathbb{R} \rightarrow \mathbb{Z}^n$.

We apply the moments accountant \cite{Abadi} and show that $\alpha_{\mathcal{DG}}(\lambda)$ can be upper bounded efficiently without evaluating the pmf of $\mathcal{DG}$.

Let $\eta_0^{\mathcal{DG}}(x|\xi) =  \mathsf{pmf}_{\mathcal{DG}(0, \xi)}(x)$ and $\eta_1^{\mathcal{DG}}(x|\xi) =  (1-q) \mathsf{pmf}_{\mathcal{DG}(0, \xi)}(x) + q \mathsf{pmf}_{\mathcal{DG}(1, \xi)}(x)$ where $\mathsf{pmf}_{\mathcal{DG}(\mu, \xi)}(x) = \Pr_{x \sim \mathcal{DG}(\mu, \xi) }[x]$.
Then,
\begin{align*}
\alpha_{\mathcal{DG}}(\lambda) &= \log\max(E_1'(\lambda, \xi), E_2'(\lambda, \xi)) 
\end{align*}

where
\begin{align*}
E_1'(\lambda,  \xi) &=  \sum_{x = -\infty}^{\infty}\eta_0^{\mathcal{DG}}(x|\xi) \cdot \left(\frac{\eta_0^{\mathcal{DG}}(x|\xi)}{\eta_1^{\mathcal{DG}}(x|\xi)}\right)^{\lambda} \\
& \leq  \frac{(1 + \kappa(\xi))^{\lambda}}{(1 - \kappa(\xi))^{\lambda+1}} \sum_{x = -\infty}^{\infty}\eta_0^{\mathcal{G}}(x|\xi) \cdot \left(\frac{\eta_0^{\mathcal{G}}(x|\xi)}{\eta_1^{\mathcal{G}}(x|\xi)}\right)^{\lambda} \\ 
& \leq  \frac{(1 + \kappa(\xi))^{\lambda}}{(1 - \kappa(\xi))^{\lambda+1}} \int_{\mathbb{R}}\eta_0^{\mathcal{G}}(y|\xi) \cdot \left(\frac{\eta_0^{\mathcal{G}}(y|\xi)}{\eta_1^{\mathcal{G}}(y|\xi)}\right)^{\lambda} dy\\
&= \frac{(1 + \kappa(\xi))^{\lambda}}{(1 - \kappa(\xi))^{\lambda+1}} E_1(\lambda,  \xi)
\end{align*}
where the first inequality follows from Lemma \ref{lem:cont_app}.
Using a similar reasoning we obtain that 
\begin{align*}
E_2'(\lambda,  \xi) &\leq \frac{(1 + \kappa(\xi))^{\lambda}}{(1 - \kappa(\xi))^{\lambda+1}} E_2(\lambda,  \xi)
\end{align*}
The theorem follows from Theorem 2 in \cite{Abadi}.
\end{proof}

\subsection{Proof of Theorem \ref{thm:fl_client}}
\label{app:proof_thm_fl_client}

As opposed to the continuous case, the sum of discrete Gaussian random variables $\sum_i \mathcal{DG}(\mu_i, \xi_i)$ does not follow the distribution of $\mathcal{DG}(\sum_i \mu_i, \sqrt{\sum_i \xi_i^2})$,  though it is very close to that if $\xi_i$ is sufficiently large.   The exact difference is quantified by the following Lemma from \cite{Micciancio017}. 

\begin{lemma}[\cite{Micciancio017}, Theorem 2.1]
\label{lem:dgs_dist}
If $\xi_i \geq\sqrt{\ln(2+2/\nu)}/\pi$ and $\bm{\mu}_i \in \mathbb{R}^n$, then 
$$
 \frac{1-\nu}{1+\nu} \leq \frac{\Pr_{\mathbf{z} \sim \sum_i \mathcal{DG}(\bm{\mu}_i, \xi_i)}[\mathbf{z}]}{\Pr_{\mathbf{z} \sim \mathcal{DG}\left(\sum_i \bm{\mu}_i, \sqrt{\sum_i \xi_i^2}\right)}[\mathbf{z}]} \leq \frac{1+\nu}{1-\nu}
$$
for any $\mathbf{z} \in \mathbb{Z}^n$
\end{lemma}

Intuitively, if $Z = \sum_{x \in \mathbb{Z}} \exp(- (x-\mu)^2/2\xi_i^2) \approx \sqrt{2\pi}\xi_i$ then
$\mathcal{DG}(\mu_i, \xi_i) \approx \mathcal{G}(\mu_i, \xi_i)$, in which case $\sum_i \mathcal{DG}(\mu_i, \xi_i) \approx \sum_i \mathcal{G}(\mu_i, \xi_i) =  \mathcal{G}(\mu_i, \sum_i\xi_i) \approx \mathcal{DG}(\mu_i, \sum_i\xi_i)$, which follows from Lemma \ref{thm:fl_client}. Indeed, it also follows from the proof of Lemma \ref{lem:cont_app} that $Z/\sqrt{2\pi}\xi_i \leq \vartheta_3(\pi \mu, \exp(-2\pi^2\xi_i^2)) \leq 1 + \frac{2e^{-2\pi^2\xi_i^2}}{1-e^{-6\pi^2\xi_i^2}} \leq 1 + \frac{2}{\exp(2\xi_i^2 \pi^2) -1} \leq  1 + \frac{2}{\exp(\xi_i^2 \pi^2) -2} \leq 1 + \nu$ which provides some insight into the condition on $\xi_i$ in Lemma \ref{lem:dgs_dist}.

For example, $\nu < 10^{-4}$ is satisfied if $\xi_i > 1$.

Let $\widehat{\mathcal{DG}}_{\xi}$ denote the distributed Gaussian mechanism which returns $\sum_{k=1}^N \mathcal{DG}(\bm{\mu}_k, \xi/\sqrt{|\mathbb{K}|}) $ where $\bm{\mu}_k \in \mathbb{R}^n$. 
The next lemma, which directly follows from Theorem \ref{thm:dg_privacy} and Lemma \ref{lem:dgs_dist},  implies Theorem \ref{thm:fl_client}.

\begin{lemma}
\label{lem:ddg_privacy}
If $\xi \geq \sqrt{|\mathbb{K}|\ln(2+2/\nu)}/\pi$, then
$
\alpha_{\widehat{\mathcal{DG}}}(\lambda|q) \leq  \alpha_\mathcal{G}(\lambda | q) + 
\log \left(\frac{(1 + \kappa(\xi))^{\lambda}}{(1 - \kappa(\xi))^{\lambda+1}} \left(\frac{1+\nu}{1-\nu}\right)^3\right)$.
\end{lemma}

\subsection{Convergence Proofs}
\label{app:conv}

All the proofs are simple adaptations of Theorem 2 from \cite{SIGNSGD_vote_robustness}.
Here we outline only the main deviations from the proof of that theorem.

\textbf{Assumptions:}
\begin{enumerate}
    \item \emph{Lower bound:} For all $x$ and some constant $f^*$, $f(x) \geq f^*$, where $f$ denotes the loss/objective function.
    \item \emph{Smoothness:} Let $g(x)$ denote the gradient of the objective function $f$ evaluated at $x$. Then, for all $x, y$ and some non-negative constant $\mathbf{L} = (L_1, L_2, \ldots, L_n)$, 
    $$|f(y) - [f(x) + g(x)^{\mathsf{T}}(y-x)]| \leq 1/2 \sum_{i} L_i(y_i - x_i)^2$$
    \item \emph{Variance bound:} Upon receiving query $x \in \mathbb{R}^n$, the stochastic gradient oracle gives us an independent, unbiased estimate $\hat{g}$ that has bounded variance per coordinate: $\mathbb{E}[\hat{g}(x)] = g(x)$, $\mathbb{E}[(\hat{g}(x)_i-g(x)_i)^2] \leq \tau_i^2$ for a vector of non-negative constants $\bm{\tau} = (\tau_1, \tau_2, \ldots, \tau_n)$. 
    \item \emph{Unimodal, symmetric gradient noise:} At any given point $x$, each component of the stochastic gradient vector $\hat{g}(x)$ has unimodal distribution that is also symmetric about the mean.
\end{enumerate}

Note that adding extra Gaussian noise to each gradient component for the purpose of differential privacy will not violate Assumption 4.

\begin{theorem}
If $|\mathbb{B}| = \Tcl$, $\Tgd=1$, and $\gamma = \sqrt{\frac{f_0 - f_*}{||\mathbf{L}||_1\Tcl}}$, then
\begin{enumerate}
\item For SignFed in the Random Update Attack, 
{\footnotesize
$$\frac{1}{\Tcl} \sum_{t=0}^{\Tcl-1} \mathbb{E}||g_t||_1 \leq \frac{2}{\sqrt{\Tcl}}\left(\frac{\sqrt{2}||\bm{\tau}||_1}{(1-\alpha)\sqrt{CN}} + \sqrt{||\mathbf{L} ||_1(f_0 - f^*)}\right)$$
}
where $\alpha$ denotes the fraction malicious clients and $|g_i|/\tau_i \leq 2/\sqrt{3}$ for all $1 \leq i \leq n$.
\item For DP-SignFed, 
{\footnotesize
\begin{align*}
\frac{1}{\Tcl} \sum_{t=0}^{\Tcl-1} \mathbb{E}||g_t||_1 \leq \frac{2}{\sqrt{\Tcl}}\left(\frac{||\bm{\tau}||_1}{\sqrt{CN}} + \frac{\sqrt{3n}\sigma||\bm{\tau}||_1}{CN} + \sqrt{||\mathbf{L} ||_1(f_0 - f^*)}\right)
\end{align*}
}
if $|g_i|/\tau_i \leq 2/\sqrt{3}$ for all $1 \leq i \leq n$.

\end{enumerate}
\end{theorem}

\begin{proof}
 The primary focus of all the proofs is to bound the probability that a client computes the sign of a parameter update correctly.
Let $M = CN$.
As in \cite{SIGNSGD_vote_robustness}, let $Z_i \in [0, M]$ denote the number of correct sign bits received by the aggregator for parameter $i$, and $p$ denotes the probability that a honest client computes the correct bit. Let $\omega = p - \frac{1}{2}$.

\begin{enumerate}
    \item \textbf{Random Update Attack:} Notice that the probability that a sign of any parameter is correct at a malicious client is $1/2$, and each client acts independently from each other. Hence,
    $\E{Z_i} = (1- \alpha)Mp + \frac{1}{2}\alpha M$
    and $\Var{Z} = \frac{1}{4}\alpha M + (1-\alpha)M p (1-p)$. 
    The probability that a vote fails for the $i^{th}$ parameter is identical to $\Prob{Z_i \leq M/2}$, which, likewise in \cite{SIGNSGD_vote_robustness},  can be bounded as follows.
    \begin{align}
        \Prob{Z_i \leq M/2} &= \Prob{\E{Z_i} - Z_i \geq \E{Z_i} - M/2} \notag \\
        &\leq \frac{1}{1 + \frac{(\E{Z_i} - N/2)^2}{\Var{Z_i}}} \tag{by Cantelli's inequality}\\
        &\leq \frac{1}{2} \sqrt{ \frac{\Var{Z_i}}{(\E{Z_i} - M/2)^2} } \tag{by $1+x^2 \geq 2x$}\\
        &\leq \frac{1}{2\sqrt{M}} \sqrt{ \frac{\frac{1}{4}\alpha + (1-\alpha)p(1-p)}{(1-\alpha)^2(p-\frac{1}{2})^2} } \notag\\
        &\leq \frac{1}{2\sqrt{M}} \sqrt{ \frac{\frac{1}{4}\alpha }{(1-\alpha)^2(p-\frac{1}{2})^2} } \notag \\
        &+ \frac{1}{2\sqrt{M}} \sqrt{ \frac{ p(1-p)}{(1-\alpha)(p-\frac{1}{2})^2} } \notag \\
        &\leq \frac{\sqrt{\alpha}}{4\sqrt{M}(1-\alpha)|\omega|} 
        + \frac{1}{2\sqrt{M}} \sqrt{ \frac{\frac{1}{4}-\omega^2}{(1-\alpha)\omega^2} } \notag \\
        &\leq \frac{\sqrt{3\alpha}\tau_i}{2\sqrt{M}(1-\alpha)|g_i|} 
        + \frac{\tau_i}{\sqrt{M(1-\alpha)}|g_i|}  \label{eq:tmp1} \\
        &\leq \frac{\tau_i(\sqrt{\alpha} + \sqrt{1-\alpha})}{\sqrt{M}(1-\alpha)|g_i|} \notag \\
        &\leq \frac{\sqrt{2}\tau_i}{\sqrt{M}(1-\alpha)|g_i|} \notag
    \end{align}
    where, in Eq.~\eqref{eq:tmp1}, we used that $\frac{1}{4\omega^2} - 1 \leq 4\tau_i^2/g_i^2$ and $1/|\omega| \leq 2\sqrt{3}\tau_i/|g_i|$ for $|g_i|/\tau_i < 2/\sqrt{3}$ based on Lemma 1 in \cite{SIGNSGD_vote_robustness}.
    The rest of the derivation is identical to the proof of Theorem 2 in \cite{SIGNSGD_vote_robustness}.
    
    \item \textbf{DP-SignFed:} The Gaussian noise is added to the sum of signs. Let $Y_i$ denote the random variable describing the noise added by the clients to $Z_i$. 
    \begin{align}
        \Prob{Z_i + Y_i \leq M/2} &\leq \frac{1}{2} \sqrt{ \frac{\Var{Z_i + Y_i}}{(\E{Z_i + Y_i} - M/2)^2} } \notag \\
        &\leq \frac{1}{2} \sqrt{ \frac{\Var{Z_i} + \Var{Y_i}}{(\E{Z_i} - M/2)^2} } \tag{by independence and $\E{Y_i} = 0$} \\
        &\leq \frac{1}{2} \sqrt{ \frac{\Var{Z_i}}{(\E{Z_i} - M/2)^2} } \notag \\ &+ \frac{1}{2} \sqrt{ \frac{\Var{Y_i}}{(\E{Z_i} - M/2)^2} } \label{eq:p}
    \end{align}
Based on \cite{SIGNSGD_vote_robustness}, 
\begin{align}
     \frac{1}{2} \sqrt{ \frac{\Var{Z_i}}{(\E{Z_i} - M/2)^2} } &\leq \frac{1}{2} \sqrt{ \frac{M\left(\frac{1}{4}- \omega^2 \right)}{M^2\omega^2} } \notag \\
     &\leq \frac{1}{2} \sqrt{ \frac{M4\tau_i / |g_i|}{M^2\omega^2} } \notag \\
     &\leq \frac{\tau_i}{\sqrt{M}|g_i|} \label{eq:p1}
\end{align}

Moreover, if $|g_i|/\tau_i \leq 2/\sqrt{3}$, then $1/\omega^2 \leq 12 \tau_i^2 / g_i^2$, and hence
\begin{align}
     \frac{1}{2} \sqrt{ \frac{\Var{Y_i}}{(\E{Z_i} - M/2)^2} } &\leq \frac{1}{2} \sqrt{ \frac{n\sigma^2}{M^2\omega^2} } \notag \\
     &\leq \frac{1}{2} \sqrt{\frac{12 n \sigma^2 \tau_i^2 / g_i^2}{M^2}} \notag \\
     &\leq \frac{\sqrt{3}\sqrt{n}\sigma\tau_i}{M|g_i|} \label{eq:p2}
\end{align}
Plugging Eq.~\eqref{eq:p1} and \eqref{eq:p2} into Eq.~\eqref{eq:p}, we obtain that the probability that the noisy vote fails for the $i^{th}$ coordinate is bounded as
\begin{align*}
        \Prob{Z_i + Y_i \leq M/2} \leq   \frac{\tau_i}{\sqrt{M}|g_i|} +   \frac{\sqrt{3}\sqrt{n}\sigma\tau_i}{M|g_i|}
\end{align*}
if $|g_i|/\tau_i \leq 2/\sqrt{3}$. The claim follows from the proof of Theorem 2 in \cite{SIGNSGD_vote_robustness}.
\end{enumerate}
\end{proof}

\subsection{Model architectures}

For MNIST and Fashion-MNIST, we use a model \cite{FedAVG} with the following architecture: a convolutional neural network (CNN) with two 5x5 convolution layers (the first with 32 filters, the second with 64, each followed with 2x2 max pooling), a fully connected layer with 512 units and ReLu activation, and a final softmax output layer. This results in 1,663,370 parameters in total.

The LFW dataset is used with a CNN of three 3x3 convolution layers (32, 64, and 128 filters, each followed with 2x2 max pooling), a fully connected layer with 256 units and ReLU activation, and a final softmax output layer with 2 units. To test the property inference attack from \cite{Property_inference}, batch size is set to 32, and the SGD learning rate is 0.01. 

The model that we use for the CIFAR dataset is called "All-CNN-C" in \cite{CIFAR_model} \cite{CIFAR_model_medium}, which consists of a CNN of 3 blocks: 
the first block has three 3x3 convolutions layers with 96 filters (the last layer has a strides of 2x2 and dropout of 0.5 is applied), the ReLu activation is used per layer. The second block has the same configuration as the previous block, except the filter size which is 192 for each layer. The last block has one 3x3 convolutions layer with 192 filters, followed by two 1x1 convolution layers: the first with 192 filters (Relu activation) and the second with 10 filters. The last layer is connected with a global average pooling layer and uses  softmax activation. We use also the Adam optimizer with a learning rate of 0.001. This results in 1,369,738 parameters in total.

Finally, we use the following model for the IMDB dataset: one embedding layer with an output size of 50 (the vocabulary size is set to 5000 and the maximum length input to 400), followed by a convolution layer of one dimension with a kernel size of 5 and 250 filters; and a max pooling layer of size 3; followed by a LSTM layer with an output size set to 70 and an output layer with one unit that uses a sigmoid activation function. We use the Adam optimizer with a learning rate of 0.001. This results in 402,701 parameters in total.

\subsection{Selection of hyperparameters}
Strictly speaking, the selection of hyperparameters in DP-SignFed, such as batch size $|\mathbb{B}|$, scaling factor $\gamma$,  or sensitivity $S$,  must also be differentially private. One option is to use public data for this purpose which comes from the same distribution as the clients' private training data. The selection of hyperparameters can also be performed using more sophisticated methods like the one in Appendix D of \cite{Abadi}. 
\subsection{Robustness of DP-SignFed against client failures}
If any client fails to add its noise share to the model update for any reason, the aggregate will not have sufficient amount of noise to guarantee differential privacy. A straightforward countermeasure is to increase the variance of the added noise so that even if $r$ clients fail, the sum of $CN-r$ noise shares are still enough for differential privacy. In particular, each client $k$ sends $\mathsf{Enc}_{K_k}(\mathcal{DG}(\mathsf{sign}(\Delta{\mbf{w}}_t^k),  \sqrt{n}\sigma\mathbf{I}/\sqrt{CN-r}))$ to the server for aggregation.
Obviously,  if less than $r$ nodes fail, the aggregate will have larger noise than what is necessary for differential privacy.

\begin{table}[!h]
	\caption{Common environment and configuration of SignFed and StdFed. $\gamma=0.001$.}
	\label{tab:fl_sign_param}
	\centering
	\begin{tabular}{c|cccc}
		\hline
		\multirow{2}{*}{Datasets}         			                & MNIST  	                                                                    & \multirow{2}{*}{IMDB}                                                                          & \multirow{2}{*}{CIFAR} \\
		                                                            & Fashion-MNIST	                                                                &                                                                                                 &                       \\
		\hline \hline
		\multirow{8}{*}{Parameters}	    & $N=1000$; 	                    & $N=1000$;    & $N=1000$; \\
		
		                                 &           $C=0.1$; &  $C=0.1$; & $C=0.1$; \\
		                                 &          $|D_k|=60$; &   $|D_k|=25$; & $|D_k|=500$; \\
		                                 &          $|\mathbb{B}|=10$; & $|\mathbb{B}|=25$; & $|\mathbb{B}|=50$; \\
		                                 &          $\Tgd=30$;          &  $\Tgd=5$;         &  $\Tgd=50$; \\
		                                 &          $\Tcl=100$;        &  $\Tcl=100$;       &  $\Tcl=400$; \\
		                                 &          $SGD$              &  $ADAM$            &   $ADAM$     \\
		                                 &          $(\eta=0.215)$     &  $(\eta=0.001)$    &  $(\eta=0.001)$ \\

		\hline
	\end{tabular}
\end{table}

\begin{table}[!h]
	\caption{Common environment of the privacy part.  $\gamma=0.005$ and $\Tcl=200$. }
	\label{tab:fl_sign_param_privacy}
	\centering
	\begin{tabular}{c|c}
		\hline
		\backslashbox{Algorithms}{Datasets}         			                & MNIST	  \& Fashion-MNIST \\
		\hline \hline
		\multirow{5}{*}{DP-SignFed \& DP-StdFed}	        	                                    & $N=6000$;  $C=1/60$;   		\\
		                                                                            & $q=1/60$;  $|D_k|=10$; \\
		                                                                            & $|\mathbb{B}|=10$;  $\Tgd=5$;  \\
		                                                                            & $SGD(\eta=0.215)$\\ & $S=1$ (DP-StdFed); \\
		                                                                            & $n=1,663,370$;\\
		             & $\delta=10^{-5}$ \\
		\hline
	\end{tabular}
\end{table}

\begin{table}[!h]
	\caption{Parameter Configuration for the Backdoor Attacks. SignFed is used with the vote aggregation $\gamma=0.001$. }
	\label{tab:fl_sign_param_backdoor}
	\centering
	\begin{tabular}{c|c}
		\hline
		\backslashbox{Attacks}{Datasets}         			                & MNIST	\\	                                
		\hline \hline
	
		\multirow{4}{*}{In-backdoor}  	                                                & $N=10$;  $C=1$;   	 \\
		                                                                                & $|D_k|=6000$;  $|\mathbb{B}|=10$; \\
		                                                                               & $\Tgd=3000$;  $\Tcl=40$;                \\
		                                                                               & $SGD(\eta=0.215)$ \\
		\hline
		\multirow{4}{*}{Out-backdoor}    	                                            & $N=10$;  $C=1$;  	 \\
		                                                                                & $|D_k|=6000$;   $|\mathbb{B}|=10$; \\
		                                                                               & $\Tgd=3000$;  $\Tcl=100$;               \\
		                                                                               & $SGD(\eta=0.215)$ \\
		\hline
		\backslashbox{Attacks}{Datasets}         			                & CIFAR \\
		\hline \hline
		\multirow{4}{*}{In-backdoor}  	                                                &   $N=10$;  $C=1$;   \\
		                                                                                & $|D_k|=50000$;   $|\mathbb{B}|=100$; \\
		                                                                                &   $\Tgd=1000$; $\Tcl=100$;  \\
		                                                                                & $ADAM(\eta=0.001)$ \\
		 \hline                                                                               
        \multirow{4}{*}{Out-backdoor}                                                   &   $N=10$;  $C=1$;  \\
                                                                                        & $|D_k|=50000$; $|\mathbb{B}|=100$; \\
                                                                                        &  $\Tgd=1000$;  $\Tcl=300$; \\
                                                                                        &  $ADAM(\eta=0.001)$ \\
				
	\end{tabular}
\end{table}

\begin{algorithm}[h]
\small
		\caption{DP-StdFed: Federated Learning with Client Privacy \label{alg:DP-StdFed}}
	\DontPrintSemicolon
	{\bf Server:}\;
	\Indp Initialize common model $w_0$\;
	\For {$t=1$ \KwTo $\Tcl$}
	{
	    Select $\mathbb{K}$ clients randomly \;
		\For {each client $k$ \textrm{in} $\mathbb{K}$}
		{	
			$\Delta \tilde{\mathbf{w}}_t^k = \mathbf{Client}_k(\mbf{w}_{t-1})$\;
		}
		$\mbf{w}_{t} = \mbf{w}_{t-1} + \frac{1}{|\mathbb{K}|} \sum_{k} \Delta \tilde{\mathbf{w}}_t^k $\;
	}
    \Indm {\bf $\mathbf{Client}_{k}(\mbf{w})$:}\;
    \Indp
	$\mbf{w}_{t-1}^k = \mbf{w}$\;
	$\Delta \mbf{w}_t^k = \mathbf{SGD}(D_k, \mbf{w}_t^{k-1}, \Tgd) - \mbf{w}_{t-1}^k$\;
	$\Delta \hat{\mbf{w}}_t^k = \Delta \mbf{w}_t^k / \max\left(1, \frac{||\Delta \mbf{w}_t^k||_2}{S}\right)$\;
    \KwOut{$\mathsf{Enc}_{K_k}(\mathcal{G}(\Delta \hat{\mbf{w}}_t^k, S \mathbf{I}\sigma /\sqrt{|K|}))$}
\end{algorithm}


\subsection{Computational Environment}

Our experiments were performed on a server running Ubuntu 18.04 LTS equipped with a Intel(R) Xeon(R) Silver 4114 CPU @ 2.20GHz, 192GB RAM, and two NVIDIA Quadro P5000 GPU card of 16 Go each. We use Keras 2.2.0 \cite{KERAS} with a TensorFlow  backend 1.12.0 \cite{TensorFlow} and Numpy 1.14.3 \cite{Numpy} to implement our models and experiments. We use  Python 3.6.5 and our code runs on a Docker container to simplify the reproducibility. 

\subsection{Datasets}
\label{sec:datasets}

The following datasets were used:
\begin{itemize}
    \item The MNIST database of handwritten digits. It consists of 28 x 28 grayscale images of digit items and has 10 output classes. The training set contains 60,000 data samples while the test/validation set has 10,000 samples \cite{MNIST} \cite{KERAS_datasets}.
    \item The CIFAR-10 dataset consists of 60000 32x32 colour images in 10 classes, with 6000 images per class. There are 50000 training images and 10000 test images. We augment the dataset to 500,000 training images by randomly shifting the original images horizontally and vertically and by randomly flipping the original images horizontally \cite{CIFAR} \cite{KERAS_datasets}.
    \item Fashion-MNIST database of fashion articles consists of 60,000 28x28 grayscale images of 10 fashion categories, along with a test set of 10,000 images \cite{Fashion-MNIST} \cite{KERAS_datasets}. 
    \item IMDB Movie reviews sentiment classification dataset of 25,000 movies reviews, labeled by sentiment (positive/negative) \cite{KERAS_datasets}. The test set contains also 25,000 movies reviews.
    \item Labeled Faces in the Wild (LFW) dataset: consists of 13,000 $62 \cdot 47$ RGB images of faces collected from the web \cite{LFW}.
\end{itemize}

\end{document}